\documentclass[conf,twocolumn]{new-aiaa}
\usepackage[utf8]{inputenc}
\usepackage{textcomp}
\usepackage{graphicx}
\usepackage[version=4]{mhchem}
\usepackage{siunitx}
\usepackage{longtable,tabularx}
\usepackage{bm}
\usepackage{color}
\usepackage{xcolor}
\usepackage{subcaption}

\usepackage{amsthm}
\newtheorem{theorem}{Theorem}[section]
\newtheorem{lemma}[theorem]{Lemma}
\newtheorem{remark}{Remark}

\newtheorem{definition}{Definition}

\usepackage{algorithm}
\usepackage{algpseudocode}
\usepackage{siunitx}
\usepackage{longtable,tabularx}
\usepackage{comment}
\usepackage{indentfirst}
\usepackage{here}
\usepackage{mathtools}
\usepackage{empheq}
\usepackage{cases}
\usepackage{cancel}
\newlength{\singlecolumnwidth}
\title{Scalable Satellite Swarm Deployment via Distance-based Orbital Transition Under $J_2$ Perturbation}
\author{Yuta Takahashi \footnote{Ph.D. Candidate, Department of Mechanical Engineering, 2-12-1 \# I3-17 Ookayama; takahashi.y.cl@m.titech.ac.jp. Student Member AIAA; Researcher and Engineer, Satellite R\&D Division, 6-3-2 Toyo.
}}
\affil{Institute of Science Tokyo, Meguro-ku, Tokyo, 152-8550, Japan
\\Interstellar Technologies Inc., Koto-ku, Tokyo, 135-0016, Japan}
\author{Shin-ichiro Sakai\footnote{Professor, Department of Spacecraft Engineering, Institute of Space and Astronautical Science, Japan Aerospace Exploration Agency, 3-1-1 Yoshinodai, Chuo-ku, Sagamihara, Kanagawa 252-5210, Japan.}}
\affil{Japan Aerospace Exploration Agency, Sagamihara, Kanagawa, 252-5210, Japan}
\begin{document}
\twocolumn[
\maketitle
\footnotetext{Presented as Paper 2025-2068 at the AIAA SCITECH 2025 Forum, Orlando, FL, January 6–10, 2025.
}
\begin{abstract}
This paper presents an autonomous guidance and control strategy for a satellite swarm that enables scalable distributed space structures for innovative science and business opportunities. The averaged $J_2$ orbital parameters that describe the drift and periodic orbital motion were derived along with their target values to achieve a distributed space structure in a decentralized manner. This enabled the design of a distance-based orbital stabilizer to ensure autonomous deployment into a monolithic formation of a coplanar equidistant configuration on a user-defined orbital plane. Continuous formation control was assumed to be achieved through fuel-free actuation, such as satellite magnetic field interaction and differential aerodynamic forces, thereby maintaining long-term formation stability without thruster usage. A major challenge for such actuation systems is the potential loss of control capability due to increasing inter-satellite distances resulting from unstable orbital dynamics, particularly for autonomous satellite swarms. To mitigate this risk, our decentralized deployment controller minimized drift distance during unexpected communication outages. As a case study, we consider the deployment of palm-sized satellites into a coplanar equidistant formation in a $J_2$-perturbed orbit. Moreover, centralized grouping strategies are presented.
\end{abstract}
]
\section*{Nomenclature}
\noindent
\begin{table}[H]
  \centering
  \renewcommand\arraystretch{1.0}
  \begin{tabular}{@{}l @{\quad=\quad} l@{}}
$C^{B/A}$ & coordinate transformation matrix from $\mathcal{A}$ to $\mathcal{B}$\\
$c_{\pm}$&$J_2$ constants $c_{\pm}=\sqrt{1\pm s_{J_2}}$\\
$k_0$&averaged $J_2$ dynamics constant $k_0=\frac{2c_+}{\omega_{xy}c_-}$\\
$k_1$&averaged $J_2$ dynamics constant $k_1=\frac{c_-\epsilon_2}{4c_+\omega_{xy}}$\\
$k_{J_2}$& 
$J_2$ coefficient $2.633e^{10} \mathrm{~km}^5 / \mathrm{s}^2$\\
$\mu_g$ & Earth gravitational parameter, 3.986e$^{14}$ m$^3$/s$^2$\\
$[\,\cdot\,]_{j}$&$j$-th element of the enclosed vector\\
$[\,\cdot\,]_{jk}$&$(j,k)$-element of the enclosed matrix\\
$\boldsymbol{x}_{jk}$&relative  vector of $j$-th and $k$-th, $\boldsymbol{x}_{j}-\boldsymbol{x}_{k}$\\
$\boldsymbol{x}_{j\leftarrow k}$&interaction vector in $j$-th by $k$-th\\
$s_{J_2}$& $s_{J_2}={k_{J_2}(1+3 \cos 2 i_{\text {ref }})}/{(4 \mu_g r_{\mathrm{ref}}^2)}$\\
 \end{tabular}
\end{table}
\section{Introduction}
\lettrine{C}ontinuous formation control realized by fuel-free actuation enables scalable distributed space structures consisting of palm-sized spacecraft swarms. These structures provide an alternative solution for large-scale space systems by lowering their technical barriers. Moreover, their payload distribution enhances system redundancy and flexibility as distributed space systems (DSS). Aerospace communities have maximized satellite size to enhance system performance, including the resolution, communication speed, and effective isotropic radiated power \cite{gardner2006james,chu2014modeling}.
These examples include space telescopes with 6.5 m apertures \cite{gardner2006james} and communication systems with 8m$\times$8m antennas \cite{tuzi2023satellite}. In particular, large-aperture space antennas can substantially reduce the size of the ground antenna and enable direct communication with small ground terminals \cite{tuzi2023satellite,takahashi2025distance,shim2025feasibility}. Future plans for space-based solar power involve deploying kilometer-scale arrays in orbit to ensure sufficient power generation \cite{mcspadden2003space}. However, despite their functionality, large-space structures face challenges such as ground testing, launch vehicle size constraints, vulnerability to single points of failure, and issues related to technology transfer. Another alternative that offers large surface areas in compact volumes is membrane structures \cite{natori1993design,sawada2011mission,you2021ka,takeda2024thermal}, as exemplified by 10 m deployable mesh antennas \cite{natori1993design} and (14 m)$^2$ solar power sails \cite{sawada2011mission}.
The alignment accuracy and scalability of these structures rely on their material properties. However, space radiation degrades support structure materials \cite{natori1993design}, leading to unpredictable deployment behavior, such as wrinkles and deformations, due to slight initial variations \cite{sawada2011mission}. Studies have also considered electrically corrected membrane antenna deployment \cite{you2021ka,takeda2024thermal} and 12 m imaging radars for forest biomass \cite{quegan2019european}. The alignment of distributed space structures requires rapid advancement of state estimation sensors; however, the behavior of rigid multi-agent systems is easier to predict relative to that of flexible structures. The in-orbit assembly strategy lost this advantage and has several drawbacks, including malfunctions and increased complexity due to the docking process. Thus, distributed space structures may soon offer larger, more precise alternatives to monolithic designs.
\subsection{Problem Statement and Related Works} 
A potential actuation method for formation control at close range is the use of fuel-free propulsion, such as satellite magnetic interaction and differential aerodynamic force. Previous studies have used magnetorquers (MTQs), or satellite-mounted coils, for multiple satellite formation control, which is referred to as electromagnetic formation flight  \cite{shim2025feasibility,takahashi2022kinematics,takahashi2021simultaneous,takahashi2020time,tajima2023study,porter2014demonstration,schweighart2006electromagnetic,ivanov2022electromagnetic}. This enables six-degree-of-freedom (6-DoF) control of an arbitrary number of satellites without consuming fuel, by managing both the electromagnetic forces and torque \cite{takahashi2022kinematics,takahashi2021simultaneous,takahashi2020time}. Previous studies have proposed angular momentum management for long-term stability \cite{zhang2016angular,takahashi2022kinematics,takahashi2021simultaneous}, a multi-leader strategy for system scalability \cite{takahashi2021simultaneous}, and learning-based approximation of the exact magnetic field model \cite{takahashi2025coil} and power-optimal current allocation \cite{takahashi2024neural} under high computational burden \cite{schweighart2006electromagnetic,tajima2023study}. Multiple experiments have provided proof-of-concept of electromagnetic formation control for ground \cite{takahashi2025experimental,takahashi2025noda_mmh} and micro-gravity environments \cite{porter2014demonstration}. As another fuel-free actuation, differential aerodynamic drag control has also been studied \cite{sun2017roto,hu2021differential,shao2017satellite,ivanov2018study,sun2018neural} based on flight demonstrations \cite{gunter2016ranging,foster2018constellation}. Studies have also designed controllable plates for satellites to decouple differential drag and lift \cite{sun2017roto} or change the center of pressure to control both translational and rotational motions simultaneously \cite{sun2018neural}. Adjusting the windward attitude also generates aerodynamic forces to achieve in-plane and out-of-plane formation control via three-axis attitude control \cite{ivanov2018study} and yaw angle control for the Earth-pointing attitude \cite{hu2021differential}. As the traditional actuation method, thrusters produce undesired plumes, whose interactions can cause one to two orders of magnitude greater acceleration than short-distance air drag \cite{fehse2003automated}, and optical contamination. We note that undesirable neighboring reaction forces by MTQ and aerodynamic force decay in smaller domains \cite{takahashi2022kinematics,sturrock2025modelling}. Therefore, fuel-free propulsion is a viable actuation method for formation control at close range as it ensures long-term stability of large-scale satellites. 

However, the scalable formation control strategy for satellite swarms remains an open problem, and fuel-free actuation suffers from its distance-dependent nature under unstable orbital dynamics. First, space demonstrations for DSSs typically use intermediate and impulsive actuators, such as thrusters. However, the control functions of these actuators are generally limited to passive formation, such as collision avoidance and bounded motion, due to operational challenges and fuel limitations \cite{morgan2012swarm,tuzi2023satellite} associated with the relative dynamic model \cite{fehse2003automated,schweighart2002high,morgan2012swarm}. The examples  include previous studies on a long-term efficient orbital controller using intermediate control \cite{morgan2012swarm}, free-flying swarm control \cite{tuzi2023satellite}, differential drag control for absolute orbital control \cite{ivanov2018study}, and bounded position control using MTQ \cite{ivanov2022electromagnetic}. In addition, the associated guidance and control strategy should realize autonomous satellite formation by a fuel-free propulsion actuation. Second, fuel-free actuation typically degrades with satellite distance. Electromagnetic forces decrease with relative distance from the actuator \cite{takahashi2022kinematics,shim2025feasibility}, making control challenging at greater distances. Aerodynamic forces also become invalid when the $J_2$-perturbed acceleration is larger than the aerodynamic drag and lift for large flow separation \cite{shao2017satellite}. Thus, we should derive an orbital transition strategy that enables satellites to maintain their control performance through connectivity with neighboring satellites rather than prioritizing time- or energy-optimal approaches.
\subsection{Contributions and Paper Organization}
This study presents an autonomous guidance and control system for satellite swarms to address the above-summarized technology gaps that prevent the realization of large-scale space structures through the continuous formation control of satellites. The specific contributions are as follows: In section \ref{Preliminaries}, the relative orbital dynamics and their closed-form solutions are formulated. In section \ref{Problem_Formulation}, the averaged $J_2$ relative orbital parameters and distance-based orbital stabilizer are derived. The derived orbital parameters represent a generalization of previous relative orbit parameters \cite{ivanov2022electromagnetic,schweighart2002high} for the $J_2$ effect. Specifically, these parameters decompose orbital motion into drift and periodic components and rewrite the $J_2$ perturbed relative orbital dynamics as a hierarchical linear system. The relative distance controller should maintain a drift of zero; however, this nonholonomic constraint prevents the realization of a constant gain state-feedback controller \cite{brockett1983asymptotic} for distance control. Motion planning also incurs high-dimensional computational costs for a small satellite. We also present a feedback controller system that minimizes both the drift motion and distance error while suppressing excessive relative distance drift motion. Conditions for control gains are also derived. In section \ref{Performance_Analysis}, numerical validation of 100 satellite swarm deployments into different orbital planes is presented. Finally, section \ref{Conclusion} presents the conclusions. Notably, this study assumes full knowledge of neighboring satellite states to demonstrate the control law effectiveness.
\section{Preliminaries}
\label{Preliminaries}
This section summarizes the mathematical method for controlling a group of satellites orbiting Earth. 
\subsection{Algebraic Graph Theory \cite{mesbahi2010graph}}
\label{Algebraic_Graph_Theory}
This subsection provides a brief overview of graph theory; for detailed information, please refer to \cite{mesbahi2010graph}. We consider an undirected graph $\mathcal{G}$ specified by $N$ nodes $\mathcal{V}=\{1,\ldots,N=|\mathcal{V}|\}$, where $|\mathcal{S}|$ denotes that the cardinality of an arbitrary set $\mathcal{S}$. For the edge set $\mathcal{E} \subseteq \mathcal{V}\times \mathcal{V}$, where $|\mathcal{E}|=P$, edge $(j, k)\in\mathcal{E}$ indicates that the $k$th node can obtain some information, such as its relative distance or interactions from the $j$th node. The sets $\mathcal{N}_j= \{k\ |\ (j, k) \in \mathcal{E}\}$ are called $j$th neighbors. A connected graph $\mathcal{G}$ includes an arbitrary spanning tree subgraph $\mathcal{G}_\tau$ and the remaining edges $\mathcal{G}_c$, that is, $\mathcal{G}=\mathcal{G}_\tau \cup \mathcal{G}_c$, and incidence matrix $E(\mathcal{G})\in\mathbb{R}^{|\mathcal{V}|\times|\mathcal{E}|}$, which is a linear combination of tree edges \cite{mesbahi2010graph}:
\begin{equation}
\label{spanning_tree}
\begin{aligned}
&[E(\mathcal{G})]_{ij}=\left\{\begin{aligned}
-1 &\text { if } v_i \text { is the tail of } e_j \\
1 &\text { if } v_i \text { is the head of } e_j \\
0 &\text { otherwise. }
\end{aligned}\right.,\\
&E(\mathcal{G})=\begin{bmatrix}
E\left(\mathcal{G}_\tau\right)& E\left(\mathcal{G}_c\right)
\end{bmatrix}
\triangleq E_\tau R(\mathcal{G})
\end{aligned}
\end{equation}
where $E_\tau\triangleq E\left(\mathcal{G}_\tau\right)$ and edges $e_{j\in[1,p]}$ are in a graph whose edges are arbitrarily oriented; $R(\mathcal{G})=[I,T_\tau^c]\in\mathbb{R}^{N-1\times P}$; and $T_\tau^c=(E_\tau^\top E_\tau)^{-1} E_\tau^\top E_c$,
such that $E(\mathcal{G}_\tau) T_\tau^c=E(\mathcal{G}_c)$. This defines the Laplacian matrix $L\in\mathbb{R}^{N\times N}$, whose eigenvalues are $0=\lambda_1(\mathcal{G})\leq \lambda_2(\mathcal{G})\leq\cdots\leq \lambda_n(\mathcal{G})$, and edge Laplacian matrix $L_\mathsf{e}\in\mathbb{R}^{P\times P}$
$$
L=E(\mathcal{G})E^\top(\mathcal{G}), \quad L_e = E^\top(\mathcal{G})E(\mathcal{G})
$$
We define the eigendecomposition of $L_{\mathsf{e}}$ and singular value decomposition of $E$ as
$$
E=U\Sigma V^\top=U_+\Sigma_+ V_+^\top, \quad L_{\mathsf{e}}=V_+(\Sigma_+^\top\Sigma_+)V_+^\top=V_+D_+V_+^\top
$$
where $U\in\mathbb{R}^{|\mathcal{V}|\times |\mathcal{V}|}$ and $V\in\mathbb{R}^{|\mathcal{E}|\times |\mathcal{E}|}$ are orthogonal matrixes, $\Sigma\in\mathbb{R}^{|\mathcal{V}|\times |\mathcal{E}|}$ and $D_+\in\mathbb{R}^{\smash{(|\mathcal{V}|-1)\times (|\mathcal{V}|-1)}}$ are the diagonal matrix, and subscript $[\cdot]_+$ indicates the submatrix corresponding to nonzero eigenvalues.
\subsection{Relative Orbital Dynamics of Spacecraft Formation Flying Under $J_2$ Gravity Effects \cite{fehse2003automated,schweighart2002high,morgan2012swarm}}
\label{Averaged_J_2_Relative_Orbital_Parameters}
This subsection describes the orbital and relative dynamics of the satellite formation flight. For the coordinate system, we mainly used the Earth-centered inertial (ECI) system $\mathcal{I}$, an orbitally fixed coordinate system, and the Local-Vertical-Local-Horizontal frame (LVLH), a $\mathcal{O}$ system. Here, $\mathcal{O}$ is spanned by the unit vectors
$$
{\mathbf{o}_x}=\mathrm{nor}(\mathbf{r}), \quad {\mathbf{o}_y}={\mathbf{o}_z} \times {\mathbf{o}_x}, \quad {\mathbf{o}_z}=\mathrm{nor}(\mathbf{h})=\mathrm{nor}(\mathbf{r} \times \dot{\mathbf{r}})
$$
where the position vectors from the Earth's center are $\mathbf{r}$, the angular momentum vector per unit mass is $\mathbf{h}$, and the normalized function is $\mathrm{nor}(x)=x/\|x\|$.

First, we introduce a reference orbit to compensate for the time-varying $J_2$ geopotential gravity effects. A previous study introduced an artificially adjusted reference orbit such that the radius $r_{\mathrm{ref}}$ is constant, and the orbital angular vector $\omega_{\mathrm{ref}}$ coincides with the average of the satellites \cite{schweighart2002high}. Let us define the $j$-th satellite position from Earth's center as $P_j=[r_{\mathrm{ref}}+x_j;y_j;z_j]$ in the $\mathcal{O}$ frame and its orbital dynamics with a gravity gradient $\nabla U_{J 2}(P_j,i,\theta)\in\mathbb{R}^3$ \cite{schweighart2002high,morgan2012swarm} as $\ddot{P}_j=\nabla U_{J 2}(P_j,i,\theta)+d_e$ where
\begin{equation}
    \label{J2 orbital dynamics}
\nabla U_{J 2}=-
\begin{bmatrix}
\frac{\mu}{\|P_j\|^2}\\0\\0
\end{bmatrix}-\frac{k_{J 2}}{\|P_j\|^4}
\begin{bmatrix}
1-3 \sin ^2 i \sin ^2 \theta\\
 \sin ^2 i \sin 2 \theta\\
 \sin 2 i \sin \theta
\end{bmatrix}
\end{equation}
where the unmodeled disturbances are $d_{e}\in\mathbb{R}^3$, such as high-order gravitational effects and unmodeled differential drag. Then, $\omega_{\mathrm{ref}}$ is computed as $\omega_{\mathrm{ref}}=
[0;0;c_+ \omega_o]$ such that
\begin{equation*}
\omega_{\mathrm{ref}}\times \omega_{\mathrm{ref}}\times r_{\mathrm{ref}}=\int_0^{2\pi}\nabla U_{J_2(\theta)}\frac{\mathrm{d}\theta}{2\pi}
\end{equation*}
where $\omega_0^2={{\mu}/{r_{\mathrm{ref}}^3}}$. Although this angular velocity adjusts the orbital period, its orbital plane is separated because of the time-varying longitudes of the ascending nodes $\Omega(t)$: $\dot{\Omega}(\theta(t))=-{2 k_{J 2} \cos i \sin ^2 \theta}/({\|\mathbf{h}\| r_{\mathrm{ref}}^3})$. To compensate for this, $\theta(t)$ is set to $\theta(t)={\omega}_{z\mathrm{ref}}t$ by using $\dot{\Omega}_{\mathrm{avg}}=\int_0^{2\pi}\dot{\Omega}(\theta)\frac{\mathrm{d}\theta}{2\pi}$ where
\begin{equation}
\label{omega_zref}
{\omega}_{z\mathrm{ref}}=c_+\omega_o -\dot{\Omega}_{\mathrm{avg}j} \cos i_j=
\omega_o \left(c_{+} + \frac{k_{J_2}\cos^2i}{\mu r_{\mathrm{ref}}^{2}}\right).
\end{equation}

Next, we introduce the approximate relative orbital dynamics of the two orbiting satellites, which are useful for deriving a closed-form solution. We define the relative position of the $j$-th satellite from the $k$th satellite as ${r_{jk}}={r}_j-{r}_k=[x_{jk};y_{jk};z_{jk}]$. Then, linearization around the reference orbit yields the dynamics of the relative motion in $\{\mathcal{O}\}$ \cite{fehse2003automated,schweighart2002high}:
\begin{equation}
\label{Hill_dynamics}
\begin{aligned}
&\begin{aligned}
&\ddot{\overline{x}}-2\omega_{xy}\dot{\overline{y}}-3\omega_{xy}^2 \overline{x}-\frac{4 \omega_{xy}^2}{c_-^2/s_{J_2}}\left(2 \overline{x}+\frac{\dot{ \overline{y}}}{ \omega_{xy}}\right)=c_+(u_x+d_x) \\
& \ddot{\overline{y}}+2\omega_{xy}\dot{\overline{x}}={c_-}(u_y+d_y)\\
&\ddot{z}+\omega_z^2  z=2 l_z \omega_z \cos (\omega_z t+\theta_z)+(u_z+d_z) \\
\end{aligned}
\end{aligned}
\end{equation}
\begin{equation*}
\left\{
\begin{aligned}
&\overline{x} =c_{+} x,\quad \overline{y} = c_- y,\quad{\omega}_{xy}=c_-\sqrt{{\mu}/{r_{\mathrm{ref}}^3}},\\
&\omega_z =\omega_{z\mathrm{ref}}+f_1(\delta \dot{\Omega}_{\mathrm{avg}jk}),\ r_z \sin \theta_z=z, \\
&l_z \sin \theta_z+\omega_z r_z \cos \theta_z=\dot{z}\\
&l_z(\delta \dot{\Omega}_{\mathrm{avg}jk})=-r_{\text {ref }}\sin i_{j} \sin i_{k}f_2(\delta \dot{\Omega}_{\mathrm{avg}jk})
\end{aligned}
\right.
\end{equation*}
where the subscript $jk$ is dropped, $f_{1,2}\in\mathbb{R}$ is a function of $\delta \Omega_{jk0}$ \cite{schweighart2002high}, and $z$-axis dynamics has added the term $2 l \omega_z \cos(\cdot)$ to compensate for the cross-track motion errors owing to time averaging. It is worth noting that this model includes the error disturbance by averaging, i.e., $\ddot{r}_{\mathrm{avg}jk}=(\nabla^2 U_{J_2}-
\int_0^{2\pi} \nabla^2 U_{J_2}\frac{\mathrm{d}\theta}{2\pi})$.
\section{Baseline Stabilizer for a Satellite Swarm in a Coplanar Equidistant Plane}
\label{Problem_Formulation}
This section develops a baseline orbital stabilizer for a satellite swarm with equal spacing in a single orbital plane. 
We first define the averaged $J_2$ relative orbital parameters to capture the averaged relative orbital motion under $J_2$ gravity. Then, this section derives its target orbital parameters along with the target plane as ``Swarm Plane $\Phi(\Theta_{P},\Theta_{z-xy})$'' and designs an orbital stabilizer to achieve a distributed space structure in coplanar equidistant. 
\subsection{Averaged $J_2$ Relative Orbital Parameters}
\label{J2_Relative_Orbital_Parameters}
First, we introduce passive stable trajectories in Earth's orbit and derive the averaged $J_2$ relative orbital parameters that are crucial for the following controller design. From the differential equation of the averaged relative orbit dynamics in Eq.~(\ref{Hill_dynamics}), we calculated the analytical solution of the future trajectory predicted at time $t=0$ under no input:
\begin{equation}
\label{CWsol}
\begin{aligned}
&\begin{bmatrix}
    {x}(t)\\
    {y}(t)\\
    {z}(t)
\end{bmatrix}
=
\begin{bmatrix}
r_{\mathrm{o}}(0,t)\\
0
\end{bmatrix}
+
\begin{bmatrix}
   r_{xy}\sin{(\omega_{xy} t + \theta_{xy})}/c_+\\
    2r_{xy}\cos{(\omega_{xy} t + \theta_{xy})}/c_-\\
    (r_z+l_z t) \sin{({\omega}_z t+\theta_z)}
\end{bmatrix}\\
&r_{\mathrm{o}}(0,t)=
\begin{bmatrix}
    {2C_{1}}&{C_4-\epsilon_{2} C_1 t}
\end{bmatrix}^\top,\ \epsilon_2 = \frac{4c_+^2-c_-^2}{c_+c_-}\omega_{xy}
\end{aligned}
\end{equation} 
where $3+5s_{J_2}=4c_+^2-c_-^2$, and its time differentiation yields Eq.~(\ref{Hill_dynamics}); $c_{\pm}$ is given by Eq. ~(\ref{Hill_dynamics}); and $r_{\mathrm{o}}(0,t)\in\mathbb{R}^2$ is the center position of the relative orbit at time $t$ predicted at $t=0$. The orbital indices computed at estimation time ($t=0$) are defined as follows:
\begin{equation}
\label{definition_C}
\left\{
\begin{aligned}
C_{1}&={c_+}/{c_-^2}(2 \overline{x}+{\dot{\overline{y}}}/{ {\omega}_{xy}}),\ C_{3}=\overline{x}-2c_+C_{1}\\
C_{4}&= (\overline{y}-{2 \dot{ \overline{x}}}/{{\omega}_{xy}})/c_-,\ C_{2}=( \overline{y}-c_-C_4)/2,\\
r_{xy}^2&= {C_2^2+C_3^2},\ \theta_{xy}= \tan^{-1}(C_3,C_2),\\ 
C_5&= \dot{z}/\omega_z,\ C_6 = z,\\
r_{z}^2&= {C_6^2+C_5^2},\ \theta_{z}= \tan^{-1}(C_6,C_5).\\
\end{aligned}
\right.
\end{equation}
As for operational convenience, the remaining ``Time'' to escape from a controllable range is sometimes a more valuable index for safety mode than distance. We provide the ``relaxed'' connectable time $\hat{T}_{\mathrm{conn}}(C_{1,4})$ for a deputy satellite to exit from a controllable region based on the analytical solution in Eq.~(\ref{CWsol}). Please refer to subsection \ref{Escape_Connectable_Time_Relaxation} in the Appendix and Fig.~\ref{Typical_relative_orbital_motion} for more accurate definitions that require computational loads. 
\begin{definition}[Connectable Time]  
We define $\hat{T}_{\mathrm{conn}}(C_{1,4})$, where a given $r_s$ coincides with the distance of the ellipse centers $\|r_{\mathrm{o}}(t)\|$, i.e., $\hat{T}_{\mathrm{conn}}(C_{1,4})=\left\{t \geq 0\ |\ \|r_{\mathrm{o}}(t)\|^2= r_s^2\right\}$ as shown in Fig.~\ref{Relaxed_connectable_period}. Its solutions are as follows: 
\begin{equation}
\label{T_conn}
\hat{T}_{\mathrm{conn}}(C_{1,4})
=\max\left(\max\left(\frac{C_4\pm \sqrt{r_s^2-(2C_1)^2}}{\epsilon_2 C_1}\right),0\right).
\end{equation}
Note that although $\hat{T}_{\mathrm{conn}}\rightarrow\infty$ as $\|r_{\mathrm{o}}(0,t)\| \rightarrow 0$, the converse is not guaranteed since $C_1=0$ yields $\hat{T}_{\mathrm{conn}}\rightarrow \infty$.
\end{definition}
\begin{figure}[tb!]
 \centering
  \begin{minipage}[b]{1\singlecolumnwidth}
    \centering
    \hspace*{-0.1cm}%
   \includegraphics[width=1.075\columnwidth]{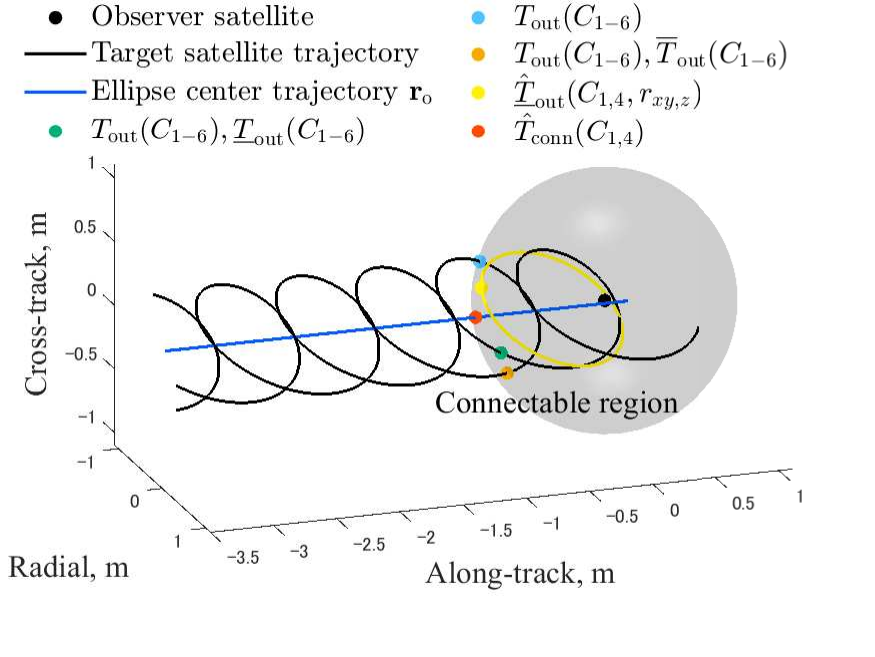}
   \vspace{-0.5cm}
    \subcaption{Typical relative orbital motion and time definitions.}
    \label{Typical_relative_orbital_motion}
  \end{minipage}\\
  \begin{minipage}[b]{1\singlecolumnwidth}
    \centering
    \hspace*{-0.25cm}%
   \includegraphics[width=1.124\columnwidth]{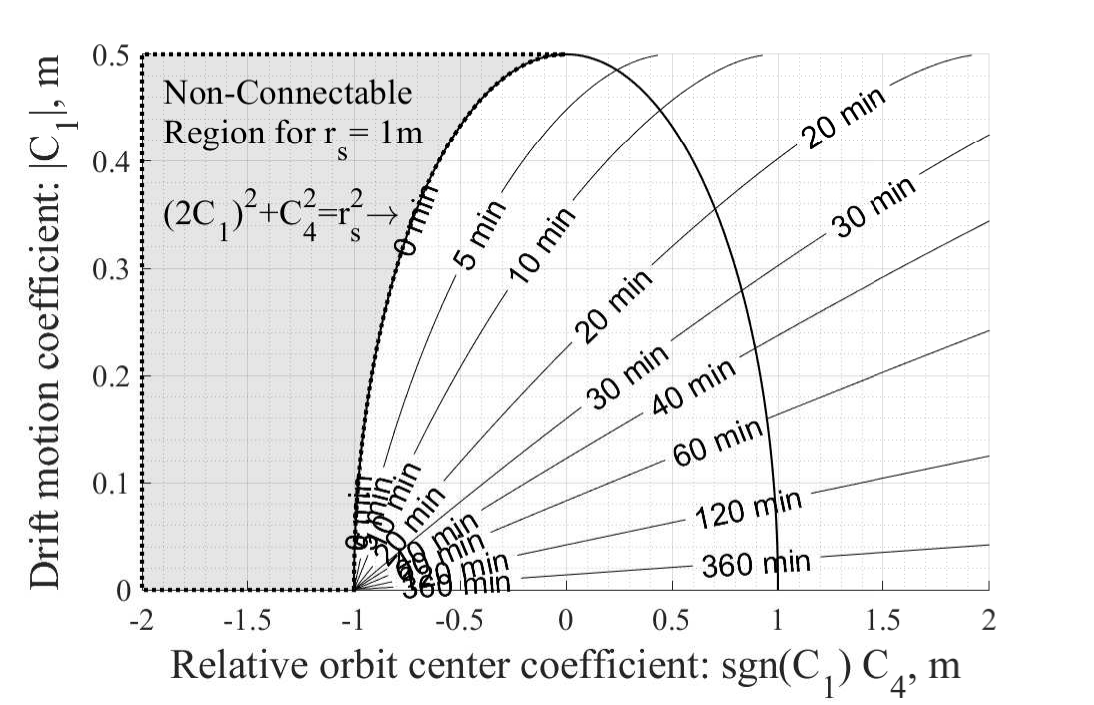}
    \subcaption{Relaxed connectable period $\hat{T}_{\mathrm{conn}}(C_{1,4})$ in Eq.~(\ref{T_conn}).}
  \label{Relaxed_connectable_period}
  \end{minipage}
  \caption{Escape, connectable, and relaxed connectable period estimates for two satellites in relative orbital dynamics derived in subsection \ref{Escape_Connectable_Time_Relaxation} in the Appendix.
  }
\end{figure}
\subsection{Target Relative Orbit Description for a Satellite Swarm with Equal Spacing in a Single Orbital Plane}
\label{swarm_plane}
Next, we derive hierarchical target values for averaged $J_2$ relative orbit parameters to achieve a satellite swarm in the coplanar equidistant plane. Since the averaged dynamics of the xy-plane and z-plane are decoupled, as shown in Eqs.~(\ref{Hill_dynamics}) and (\ref{CWsol}), we set hierarchical targets for $z$-axis motion. In addition, the target states are divided into three types: 1) drift elimination in the xy-plane, 2) orbit plane control in the z-plane, and 3) distance control in the xy-plane.
\subsubsection{Drift Elimination in Swarms with Identical Inclinations}
First, the condition for drift elimination is given as $(C_1,l_z,C_4)=(0,0,0)$ by analytical solution in Eq.~(\ref{J2_Relative_Orbital_Parameters}). Here, $(C_1,l_z)=(0,0)$ for an arbitrary pair in a satellite swarm that achieves a passively stable close orbit, that is, fuel-optimal trajectories with no relative drift motion \cite{fehse2003automated,schweighart2002high,morgan2012swarm,ivanov2022electromagnetic}. Along with these constraints, $C_4=0$ introduces a collision avoidance trajectory due to concentration \cite{morgan2012swarm}. Moreover, the uniformity of the difference in $\theta_z,\theta_{xy}$ introduces ideal orbital planes. Here, $l_z=0$ is naturally satisfied if the satellites have identical inclinations such that $\delta i_{jk}=i_j - i_k \approx 0$ for arbitrary $j$th and $k$th inclinations. For $\delta i_{jk}\approx 0$, the $|\dot{\Omega}_{\mathrm{avg}}|$ in Eq.~(\ref{omega_zref}) is bounded as $|\dot{\Omega}_{\mathrm{avg}}|\lesssim 2e^{-6}\cos i_{\mathrm{ref}}$, which derives 
$$
|\delta \dot{\Omega}_{\mathrm{avg}jk}|
\approx 2e^{-6}\sin i_k \frac{\dot{z}_{jk0}}{\omega_{z\mathrm{ref}} r_{\text {ref}}}
\lesssim 2.5e^{-10}\dot{z}_{jk0}
$$
by Taylor series expansion $
\cos(i_k + \delta i_{jk})\approx \cos i_k-\sin i_k\cdot\delta i_{jk}$. The definitions in Eq.~(\ref{Hill_dynamics}) and $|\delta \dot{\Omega}_{\mathrm{avg}jk}|\approx 0$ yield 
$$
f_{1,2}(\delta \dot{\Omega}_{\mathrm{avg}jk})\approx 0\ \Rightarrow\ \omega_z \approx\omega_{z\mathrm{ref}},\quad l_z(\delta \dot{\Omega}_{\mathrm{avg}jk})\approx 0.
$$
These results show that 
The $z$-axis amplitude takes a constant value, and the orbital plane rotation due to $\omega_x\neq \omega_z$ presents a steady rate under $\delta i_{jk}\approx 0$. Note that the validity of $\delta i_{jk}\approx 0$ can be evaluated by $\delta i_{jk}={\dot{z}_{jk0}}/{(\omega_{\mathrm{ref}z} r_{\text {ref}}})$ \cite{schweighart2002high}. 

Note that for the non-drift conditions between the reference orbit and arbitrary $j$th satellite, the following constants are added \cite{schweighart2002high}: $$\frac{c_+}{c_-^2}\left(\dot{\overline{y}}_j+\frac{2 \overline{x}_j}{ \omega_{xy}}\right)=C_{1c}=\frac{c_+}{c_-^2}\frac{c_- k_{J_2}\sin ^2 i_{\mathrm {ref }}}{2\omega_{z\mathrm{ref}}r^4_{\mathrm{ref}}}.
$$
\subsubsection{Relative Orbit Plane Control}
\begin{figure}[tb!]
       \centering
\includegraphics[width=\singlecolumnwidth]{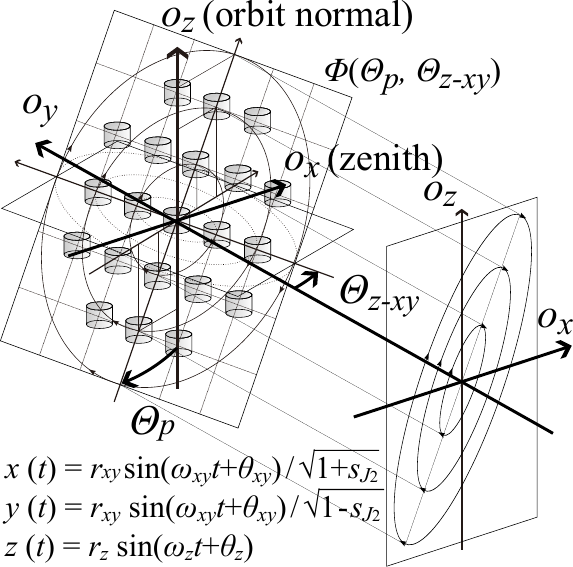}
\caption{One desired example of the distributed space structure on the grid.}
\label{grid_formation_ver2}
\end{figure}
As the second target, we describe relative orbit plane control in the z-plane. The analytical orbital solutions in Eq.~(\ref{CWsol}) show the conditions for overall satellite swarms to stay in the same orbital plane: $\delta \theta=(\theta_z-\theta_{xy})$ and $(\omega_{xy}-\omega_z)$ take the same values between arbitrary satellite pairs. We can simply calculate specific values of $\delta \theta$ and apply user-defined orbital plane angles $(\Theta_{P},\Theta_{z-xy})$ as follows:
\begin{equation}
\begin{aligned}
&\delta\theta(t)=\theta_{z}(t)-\theta_{xy}(t)=\tan^{-1}(2\tan \Theta_{z-xy}),\\
&r_{z}(t)= \frac{r_{xy}(t)}{\tan\Theta_P}\frac{\cos\Theta_{z-xy}}{\cos(\delta\theta(t))}
\end{aligned}
\end{equation}
Subsequently, the hierarchical target of $r_{zd}$ and $\omega_{zd}$ in Eq.~(\ref{definition_C}) for a given $r_{xy}$, $\theta_{xy}$, and $\omega_{xy}$ can be defined along with $C_{5d0,6d0}$ as follows:
\begin{equation}
\label{r_xy_and_r_z}
\begin{aligned}
&\left\{
\begin{aligned}
&\theta_{zd}(\Theta_{z-xy})\triangleq \theta_{xy}(t)+\tan^{-1}(2\tan \Theta_{z-xy}),\\
&r_{zd}\triangleq \frac{r_{xy}(t)}{\tan\Theta_P}\frac{\cos\Theta_{z-xy}}{\cos(\theta_z - \theta_{xy})},\quad \omega_{zd}\triangleq\omega_{xy}(i(t))
\end{aligned}
\right.\\
&\Leftrightarrow C_{5d}=r_{zd} \cos(\theta_{zd}),\quad C_{6d} =r_{zd} \sin(\theta_{zd})
\end{aligned}
\end{equation}
Note that $\omega_{zd}=\omega_{xy}$ is achieved by active $z$-axis orbit control and the feedforward term of $z$-axis $u_{dz}(t)\in\mathbb{R}$, thereby maintaining the $z$-axis orbital frequency $\omega_z$ into $\omega_{xy}$: 
\begin{equation}
\label{z_axis_feedforward_term}
u_{dz}=-r_{zd}\omega_{z}^2 \sin(\omega_z t + \theta_z)
+r_{zd}\omega_{xy}^2 \sin(\omega_{xy} t + \theta_z).
\end{equation}
These yield the desired stable trajectories $p_d(t)$ for realizing the described satellite swarm trajectories as shown in Fig.~\ref{grid_formation_ver2} as follows:
\begin{equation}
\label{desired_position}
{p}_d=
\begin{bmatrix}
   (1/c_+)r_{xyd}\sin{(\omega_{xy} t + \theta_{xy})}\\
    (1/c_-)2r_{xyd}\cos{(\omega_{xy} t + \theta_{xy})}\\
    \frac{r_{xyd}}{\tan\Theta_P}\frac{\cos(\Theta_{z-xy})}{\cos(\theta_z - \theta_{xy})} \sin{({\omega}_{xy} t+\theta_z(\Theta_{z-xy}))}
\end{bmatrix}
\end{equation}
\subsubsection{Distance Adjustment on Normalized Swarm Frame}
The third objective is to realize distance adjustment on a virtual array to achieve an almost constant relative distance related to the user-defined average distance ${r}_{\mathrm{avg}}$. Although the distance of orbital stable solutions $p_d(t)$ in Eq.~(\ref{desired_position}) always takes a periodic value, a constant orbital distance $r_{xy,z}$ is allowed for a stable trajectory. Then, our controller is applied to ensure that the average distance satisfies the user-defined values, i.e., $r^2_{\mathrm{avg}}=\frac{2\pi}{\omega_{xy}}\int_0^{{2\pi}/{\omega_{xy}}} \|{p}_d(\tau)\|^2{\mathrm{d}\tau}$. Then, the desired $r_{xyd}$ is derived based on $r_{zd}(t)$ in Eq.~(\ref{r_xy_and_r_z}):
\begin{equation}
\label{r_avg}
\begin{aligned}
&r_{xyd}(r_{\mathrm{avg}})\triangleq r_{\mathrm{avg}}/S(\Theta_{P},\Theta_{z-xy},s_{J_2})\\
&S=\sqrt{\frac{1}{2}\left(\frac{1+3 \sin^2\Theta _{z-xy}}{{\tan^2\Theta_P}}+\frac{3s_{{J_2}}+5}{1-s_{J_2}^2}\right)}
\end{aligned}
\end{equation}
Note that periodic $p_d(t)$ leads to time-variant neighbors, even around $p_d$, which prevents convergence of the distance-based formation control. To avoid this, we define a normalized swarm frame $\{\hat{\mathcal{S}}\}$ for time invariant neighbors, and it is set to satisfy $\|C^{\hat{S}/O}p_d(t)\|$ = const. This unique frame is obtained by the constraint $\{\hat{\mathcal{S}}\}$, such that its y-z plane aligns with $\Phi(\Theta_{P},\Theta_{z-xy})$, i.e., 
$$
\left\|C^{\hat{S}/O}p_d(t)\right\|=\left\|2r_{xyd}
\begin{bmatrix}
0\\
c_{(\omega_{xy}t+\theta_{xy})}\\
s_{(\omega_{xy}t+\theta_{xy})}
\end{bmatrix}
\right\|=2r_{xyd}.
$$ 
These simple calculations are used to derive the coordinate transformation matrix $C^{\hat{S}/O}(\Theta_{P},\Theta_{z-xy},c_{\pm})$, where $C^{\hat{S}/O}=C^{\hat{S}/S} C^{{S}/O_{J_2}}C^{O_{J_2}/O}$, the coefficient matrix $C^{O_{J_2}/O}$$=\mathrm{diag}([c_+,c_-,1])\in\mathbb{R}^{3\times 3}$ compensates for the $J_2$ term and the matrix $C^{\hat{S}/S}$ adjusts the distance of the elliptical orbit in the y-z plane of frame $\{{\mathcal{S}}\}$:
\begin{equation*}
\label{conversion_for_grouping}
\begin{aligned}
&C^{\hat{S}/S}=
\begin{bmatrix}
1&0&0\\
0&1&0\\
0&0&2
\end{bmatrix}
\begin{bmatrix}
1&0&0\\
0&c_{\Theta_{z-xy}}&s_{\Theta_{z-xy}}\\
0&-s_{\Theta_{z-xy}}&c_{\Theta_{z-xy}}
\end{bmatrix}
\begin{bmatrix}
1&0&0\\
0&1&0\\
0&0&{s_{\Theta_P}}
\end{bmatrix}\\
&C^{\mathcal{S}/O_{J_2}}= 
\begin{bmatrix}
{c}_{\Theta_P}&0&-{s}_{\Theta_P}\\
    0&1&0\\
    {s}_{\Theta_P}&0&{c}_{\Theta_P}
    \end{bmatrix}
\begin{bmatrix}
{c}_{\Theta_{z-xy}}&{s}_{\Theta_{z-xy}}&0\\
-{s}_{\Theta_{z-xy}}&{c}_{\Theta_{z-xy}}&0\\
0&0&1
\end{bmatrix}
\end{aligned}
\end{equation*}
Note that $C^{\mathcal{\hat{S}}/O}$$[-1;0;0]$ indicates the nadir direction used for link budget analysis and aperture design. $C^{\hat{S}/O}$ is used for our centralized grouping algorithm. We also note that we can obtain the time-invariant positional relationship in the frame $\{\overline{S}\}$ derived by $C^{\overline{S}/\hat{S}}C^{\hat{S}/O}$ to compensate for the rotation induced by orbital angular velocity $\omega_{xy}$:
$$
C^{\overline{S}/\hat{S}}(t)=
\begin{bmatrix}
1&0&0\\
0&\cos{\omega_{xy}t}&\sin{\omega_{xy}t}\\
0&-\sin{\omega_{xy}t}&\cos{\omega_{xy}t}
\end{bmatrix}.
$$
\subsection{Distance-based Orbital Stabilizer}
\label{orbit_control}
This subsection designs a distance-based orbital stabilizer to stabilize the satellite swarm into a constant geometric virtual array, as characterized in the previous subsection~\ref{swarm_plane}. We first extend the orbital parameters for two satellites into $N$ parameters. Let $[{C}_{i\in[1,6]}]$ be defined as the averaged $J_2$ orbital parameter vector for $N$ satellites:
$$
[{C}_i]=[C_{i(1c)};\ldots; C_{i(Nc)}]\in\mathbb{R}^N
$$ where $C_{i(jc)}$ represents each parameter of the $j$th satellite with respect to the formation center mass. Note that information on formation center mass is not required in the following steps, and only relative information is needed. The system equation is as follows:
\begin{equation}
\label{dot_C14}
\begin{aligned}
\begin{bmatrix}
    [\dot{C}_1]\\ [\dot{C}_4]\\
    [\dot{C}_2]\\ [\dot{C}_3]\\
    [\dot{C}_5]\\
    [\dot{C}_6]
\end{bmatrix}&=
\begin{bmatrix}
   0\\
     -\epsilon_2 [C_1]\\
    -\omega_{xy} [C_3]\\
    \omega_{xy} [C_2]\\
    -\omega_z [C_6]\\
    \omega_z [C_5]
\end{bmatrix}
+
\begin{bmatrix}

\begin{bmatrix}
0&\frac{k_0}{2}\\
-k_0& 0\\
\frac{c_-k_0}{2}&0\\
0&-c_+k_0
\end{bmatrix}
\begin{bmatrix}
    U_x+[{d}_{x}]\\
    U_y+[{d}_{y}]
\end{bmatrix}\\
 \frac{1}{\omega_z}(U_z+[{d}_z])\\
    0
\end{bmatrix}
\end{aligned}
\end{equation}
where the constant $k_0=\frac{2c_+}{\omega_{xy}c_-}\in\mathbb{R}$, $[u_{x,y,z}]\in\mathbb{R}^N$ and $[d_{x,y,z}(t)]\in\mathbb{R}^n$ are each-axis input and perturbation vectors for $N$ satellite. This also includes the averaging error of $J_2$ gravity effects in subsection~\ref{Averaged_J_2_Relative_Orbital_Parameters} and feedforward term $z$-axis $u_{dz}(t)$ in Eq.~(\ref{z_axis_feedforward_term}).

We also derive a baseline networked controller $f_j=m_ju_j$ for the linear time-invariant system in Eq.~(\ref{dot_C14}). The relative distance controller should prevent excessive relative distance drift, i.e., $(C_1, C_4)\approx(0,0)$. The nonholonomic constraints $(C_1, C_4)=(0,0)$ prevent the generation of a state-feedback controller \cite{brockett1983asymptotic} for distance control. Accordingly, a Lyapunov-based approach is employed to simultaneously suppress drift behavior and steer the relative distance toward the desired setpoint. Based on this framework, a stabilizing control law is developed, along with sufficient conditions for selecting appropriate control gains. We define the amplitude errors $s(r_{xyjk})\in\mathbb{R}$ on $\mathbf{o}_{xy}$ plane:
$$
s(\mathsf{r})= 
\frac{\mathsf{r}-r_{xyd}}{\mathsf{r}-\underline{r}_{xy}}+\frac{\mathsf{r}-r_{xyd}}{(2r_{xyd}-\underline{r}_{xy})-\mathsf{r}}
$$
where $\underline{r}_{xy}$ is related to the maximum length of the satellites, and $(r_{xyjk}-2a_s)^{-1}$ is related to collision avoidance, a fundamental requirement for proximity operations. We also define the distance function $\overline{r}_{xyjk}^2\in\mathbb{R}$ between arbitrary $j,k$-th satellites:
\begin{equation}
\label{overline_r_xy}
{\overline{r}}_{xy(\varrho,\varsigma,\tau,\psi)}^2\triangleq (\varrho C_1+C_3)^2+ \tau(\varsigma (C_4-\psi C_1)+C_2)^2
\end{equation}
where we omit the index $jk$ for simplicity. Then, we design our main controller $u_j$ for $j$th satellte using the Lyapunov direct method \cite{khalil2002nonlinear} and artificial potential function \cite{olfati2006flocking,dimarogonas2008stability} as follows:
\begin{equation}
\label{networked_input_i}
\begin{aligned}
{u}^{\mathrm{main}}_j=&\sum_{k\in\mathcal{N}_j}
\begin{bmatrix}
\frac{k_A}{2k_0}
\begin{bmatrix}
    \gamma_A \gamma_B^2\left( [C_4]-\left(\psi-\frac{2f_{0}}{\gamma_A \gamma_B^2}\right)[C_1]\right)\\
    -[2C_1]\\
\end{bmatrix}\\
-k_{z}{\omega_z}[C_{5}-C_{5d}]
\end{bmatrix}_{jk}\\
&+\sum_{k\in\mathcal{N}_j}
\begin{bmatrix}
\frac{c_-}{2k_0k_1^2} k_{B} s({\overline{r}}_{xy(2c_+,0,1,\psi)})
\begin{bmatrix}
     -\gamma_B^2 C_{2}\\
    k_1 g_{0}C_{2}
\end{bmatrix}\\
0
\end{bmatrix}_{jk}
\end{aligned}
\end{equation}
where $k_1\triangleq (c_-\epsilon_2)/(4c_+\omega_{xy})$. We denote the competing control input ${u}^{\mathrm{opp}}_j\in\mathbb{R}^3$, which is constructed based on an opposing strategy and serves as a baseline for comparison:
\begin{equation}
\label{sub_networked_input_i}
\begin{aligned}
{u}^{\mathrm{opp}}_j=&\sum_{k\in\mathcal{N}_j}
\begin{bmatrix}
\frac{k_A}{2k_0}
\begin{bmatrix}
    \gamma_A[C_4]\\
    -[2C_1]+\frac{\gamma_A \epsilon_2}{k_A}[C_4]\\
\end{bmatrix}\\
-k_{z}{\omega_z}[C_{5}-C_{5d}]
\end{bmatrix}_{jk}\\
&+\sum_{k\in\mathcal{N}_j}
\begin{bmatrix}
\frac{c_-}{2k_0k_1^2} k_{B} s({\overline{r}}_{xy(2c_+,0,1,0)})
\begin{bmatrix}
     -k_1^2 C_{2}\\
    -k_1 \frac{\epsilon_2}{k_A}C_{2}
\end{bmatrix}\\
0
\end{bmatrix}_{jk}
\end{aligned}
\end{equation}
The behaviors of Eq.~(\ref{dot_C14}) applied ${u}^{\mathrm{main}}_j$ in Eq.~(\ref{networked_input_i}) and ${u}^{\mathrm{opp}}_j$ in Eq.~(\ref{sub_networked_input_i}) are shown in the next theorem. 
\begin{theorem}[Distance-based Orbital Stabilizer]
\label{theorem_networked_controller}
For $N=|\mathcal{V}|$ satellites, we assume that their connected graph $\mathcal{G}(\mathcal{V},\mathcal{E})$ and its incidence matrix $E\in\mathbb{R}^{N\times |\mathcal{E}|}$ are given and the controllable region is sufficiently large, i.e., $r_s\gg 1$. Then, ${u}^{\mathrm{main}}_j$ in Eq.~(\ref{networked_input_i}) is applied to the un-perturbed relative orbital dynamics in Eq.~(\ref{dot_C14}), where
\begin{equation}
    \label{coefficient_condition}
\left\{
\begin{aligned}
&\psi
\triangleq\frac{-\lambda_0 \pm \sqrt{\lambda_0^2-4\gamma_B^2}}{k_1/(1-k_1)},\quad \lambda_0\geq 2\gamma_B \geq 0\\ 
&f_{0}\triangleq\frac{\psi}{2}-\lambda_0,\quad
g_{0}
\triangleq
-k_1f_{0}-\lambda_0
\end{aligned}
\right.
\end{equation}
acheives $E^\top [C_4]\rightarrow 0$ and $E^\top U_x\rightarrow 0$ as $t\rightarrow \infty$ if the nonzero eigenvalues of the Laplacian satisfy $\lambda_0 \geq 0$
\begin{equation}
\label{definition_d}
\frac{\epsilon_2/k_A}{\lambda_0+2\gamma_B}I\preceq D_+ \preceq\frac{\epsilon_2/k_A}{\lambda_0-2\gamma_B}I.
\end{equation}
For $k_B=0$, the relative error of drift coefficient $E^\top[C_1]$ is exponentially decayed and $E^\top[C_4]\rightarrow 0$ as $E^\top[C_1]\rightarrow 0$ for arbitrary gains. Moreover, ${u}^{\mathrm{opp}}_j\in\mathbb{R}^3$ in Eq.~(\ref{sub_networked_input_i}) yields globally asymptotically stable results, i.e., $E^\top[C_{1,4}]\rightarrow 0$, $W_sE^\top[C_2]\rightarrow 0$ as $t\rightarrow\infty$, for arbitrary gains. For both inputs ${u}^{\mathrm{main}}_j$ and ${u}^{\mathrm{opp}}_j$, $\overline{r}_{xyjk}\rightarrow \text{const}$ as $t\rightarrow 0$, and it satisfies $W_sE^\top[C_2]=0$ for $k_B\neq 0$. Then, $(C_5,C_6)$ asymptotically converges into $(C_{5d},C_{6d})$ for unperturbed dynamics.
\end{theorem}
\begin{proof}
See subsection \ref{proof_theorem_networked_controller} in the Appendix.
\end{proof}
\begin{remark}
When the minimum/maximum graph degree $\delta$/$\Delta$ are given, the obvious option of $\lambda_0$ and $\gamma_B$ in Eqs.~(\ref{coefficient_condition}) and (\ref{definition_d}) are as follows:
\begin{equation}
    \label{definition_d2}
\left\{
\begin{aligned}
\lambda_0&=\frac{\epsilon_2}{2k_A}\left(\frac{|\mathcal{V}|}{2\delta} +\frac{1}{2\Delta}\right)\\    
\gamma_B&=\frac{\epsilon_2}{4k_A}\left(\frac{|\mathcal{V}|}{2\delta}-\frac{1}{2\Delta}\right)
\end{aligned}
\right.
\Rightarrow
\left\{
\begin{aligned}
\frac{2\delta}{|\mathcal{V}|}&=\frac{\epsilon_2/k_A}{\lambda_0+2\gamma_B}\\    
2\Delta &=\frac{\epsilon_2/k_A}{\lambda_0-2\gamma_B}
\end{aligned}
\right.
\end{equation}
Since the nonzero eigenvalues of $L$ satisfies $\min_i[D_+]_{ii}\geq {2\delta}/{|\mathcal{V}|}$ and $\max_i[D_+]_{ii}\leq 2\Delta$, $\lambda_0$ and $\gamma_B$ in Eq.~(\ref{definition_d2}) guarantees the inequality in Eq.~(\ref{definition_d}).
\end{remark}
\noindent
To enforce the maximum graph degree $\Delta$ for satellite grouping, we next formulate the directed multi-leader graph with multi-leaders and followers.
\begin{definition}[Multi-Leader Grouping and Graph]
\label{multi_leader_graph}
For a given agent set $\mathcal{V}$, we assign its specific subset as the leader subset $\mathcal{V}_l\in\mathcal{V}$ and define the $l$th leader with followers as the “$l$th Local Group” $\mathfrak{g}_{l}$. Assuming all agents in $\mathfrak{g}_{l}$ interact with each other, an undirected graph $\mathcal{G}(\mathcal{V}, \mathcal{E})$ is defined, and the associated adjacency matrix is as follows:
$$
[A(\mathcal{G})]_{j k}=
\begin{cases}1
& \text { if }k,j\in\mathfrak{g}_{l}\quad \forall l\in\mathcal{V}_l\\
0 & \text { otherwise }\end{cases}.
$$
where the maximum number of other satellites' interactions for an arbitrary $\Delta$ is defined as follows: 
$$
\Delta=\max_j[\operatorname{Diag}(A(\mathcal{G}) \mathbf{1})]_{jj}
$$
We assume that $\overline{n}_{f\leftarrow l}$ indicates the maximum follower number for one local group and $\overline{n}_{l\leftarrow f}$ indicates the maximum following number for each satellite, which means that each satellite belongs to at most $(\overline{n}_{l\leftarrow f}+1)$. These bounds provide $\Delta$ as a function of $\overline{n}_{l\leftarrow f},\overline{n}_{f\leftarrow l}$ 
e.g. $\Delta(\overline{n}_{l\leftarrow f},\overline{n}_{f\leftarrow l})=(\overline{n}_{l\leftarrow f}+1)\overline{n}_{f\leftarrow l}\ $ if $\ \overline{n}_{l\leftarrow f}=\overline{n}_{f\leftarrow l}=2$.
$$
\sum_{l\in\mathcal{V}_l}\left( \sum_{k,j\in\mathfrak{g}_{l}}1\right)\triangleq \sum_{l\in\mathcal{V}_l}\frac{\Delta}{(\overline{n}_{l\leftarrow f}+1)}
=\Delta
$$
\end{definition}
\section{Performance Analysis: Satellite Swarm Array Deployment in $J_2$-Perturbed Orbit}
\label{Performance_Analysis}
This section utilizes numerical simulations to evaluate the developed distance-based orbital controller for 50 and 100 satellites in a $J_2$-perturbed Earth orbit. In these simulations, $k_A = 1.5e^{-2}$ is set based on the orbital time constant \cite{takahashi2025experimental}, the average distance $r_{\mathrm{avg}}$ is $0.5$, the controllable distance $r_s$ is $1$ m, and the maximum magnitude of force $\overline{f}$ is $5e^{-1}\mu$N. In addition, the satellite shape is cubic with a side length $1e^{-1}$ m, and the mass is 0.5 kg. Moreover, their initial orbital inclinations are identical at 51.7$^\circ$. The simulations in this study set $\overline{n}_{l\leftarrow f}=5$,\ $\overline{n}_{f\leftarrow l}=5$ in the definition ~\ref{multi_leader_graph} to avoid excessive interactions. For each simulation, one reference satellite $S/C_{\mathrm{ref}}$ is randomly chosen, and its position is set as the origin of the orbital frame $\{\mathcal{O}\}$. The initial orbital values $C_{2,3,5,6}$ of other satellites are randomly chosen as follows:
$$
|C_{2,3,5,6}|\leq \frac{r_s}{10}\quad\mathrm{s.t.}\quad\|\overline{p}_{(0)}-r_{\mathrm{o}(0)}\|^2 \leq 4r_{xy}^2+r_z^2\leq r_s^2.
$$
Additional orbit values $(C_{1,4},\Theta_{P,Z-XY})$ are set in each subsection.
Neighboring satellites are updated by the custom-made centralized grouping for a connected graph / subgraph in algorithm \ref{alg:Centralized_Multi_Leader_Grouping} in Appendix \ref{Appendices}. Any grouping algorithm can replace this, although we highly recommend guaranteeing that they converge to the time-invariant neighboring pairs for uniform formation. We execute the algorithm \ref{alg:Centralized_Multi_Leader_Grouping} based on the update scheduler. The adjacency matrix is updated at increasing intervals, specifically every 5, 10, 20, and 40 min, which correspond to the elapsed time of $T_{\mathrm{end}}/4$ for simulation time $T_{\mathrm{end}}$. Note that the grouping period at the initial phase should not coincide with the orbital period (approximately 100 minutes). Otherwise, the initial distorted formation may persist as the final configuration, which can degrade beamforming performance. Compared with terrestrial actuators, most spacecraft actuators exert reaction disturbance on neighbors due to linear momentum conservation. Since this reaction disturbance could be the main disturbance for ultra-close applications, we provide the convex optimization-based multi-leader grouping as well.
\subsection{Connectable Period Maximization Results}
\label{Connectable_Period_Maximization}
We first show the fault tolerance of ${u}^{\mathrm{main}}_j$ in Eq.~(\ref{networked_input_i}) to satellite connection losses using the $\hat{T}_{\mathrm{conn}}$ index. We compare the convergence performance result on the $C_4$-$C_1$ plane for the two closed loop systems applied the input ${u}^{\mathrm{main}}_j$ and ${u}^{\mathrm{opp}}_j$ in Eqs.~(\ref{networked_input_i},\ref{sub_networked_input_i}) with fixed grouping. Based on the results of theorem~\ref{theorem_networked_controller}, we chosen gain values for $k_B=0$: $f_{0}=-{\epsilon_2}/{(2k_A)}$, $\psi=0$, $\gamma_A=\gamma_B=1$, $g_{0}=0$ to make the control input order as similar as possible. We then set the initial orbital parameters for each term of $[C_4]$ as $|C_4|\leq 2.5r_s$ and two cases $C_1$ as follows:
$$
\text{Large $C_1$: }|C_1|\leq r_s/2,\quad \text{Small $C_1$: }|C_1|\leq (r_s/2)/20.
$$
The history $\hat{T}_{\mathrm{conn}}$ from the reference satellite S/C$_{\mathrm{ref}}$ in the closed-loop systems are summarized in Fig.~\ref{Relaxed_connectable_period_for_two_cases}. $C_{1,4}$ is converged into the origin for both cases, although the value of $\hat{T}_{\mathrm{conn}}(C_{1},C_4)$ differed for each result. For the case of the large $C_1$ in Fig.~\ref{Relaxed_connectable_period_for_two_cases} a), the two control trajectories are almost identical; however, ${u}^{\mathrm{main}}_j$ decreases with a consistent sign, while some trajectories by ${u}^{\mathrm{opp}}_j$ show a tendency for $C_1$ to increase. This tendency is especially evident in the case of small $C_1$ values shown in Fig.~\ref{Relaxed_connectable_period_for_two_cases} b). According to Eq.~(\ref{dot_C14}), the dynamics of $C_4$ with small $C_1$ approach a pure integrator system; however, ${u}^{\mathrm{opp}}_j$ increases $C_1$, leading to a trajectory where $T$ becomes large. Moreover, in the $u_{\max}$ plot, $C_1$ converges to the origin without increasing, and the trajectory changes in the direction of increasing $\hat{T}_{\mathrm{conn}}$. This is because the $x$-component of $\hat{T}_{\mathrm{conn}}$ includes cross terms of $C_1$ and $C_4$, which increases $\hat{T}_{\mathrm{conn}}$ as shown in Fig.~\ref{Relaxed_connectable_period_for_two_cases} c). From the definition, $\hat{T}_{\mathrm{conn}}(C_{1},C_4)$ suggests the escape time when relative orbital center escape from the controllable region of observer satellite: when satellites 1) implement each input from initial time until an arbitrary time $t$ and 2) show free motion for some unintended reasons from time $t$ and escape at $t+\hat{T}_{\mathrm{conn}}$. Therefore, the $u_{\mathrm{main}}$ trajectory yields more fault tolerance results relative to satellite connection losses under the distance constraint of $r_s$.
\begin{figure}[tb!]
\centering
\begin{minipage}[b]{1\singlecolumnwidth}
    \centering
\includegraphics[width=1.075\columnwidth]{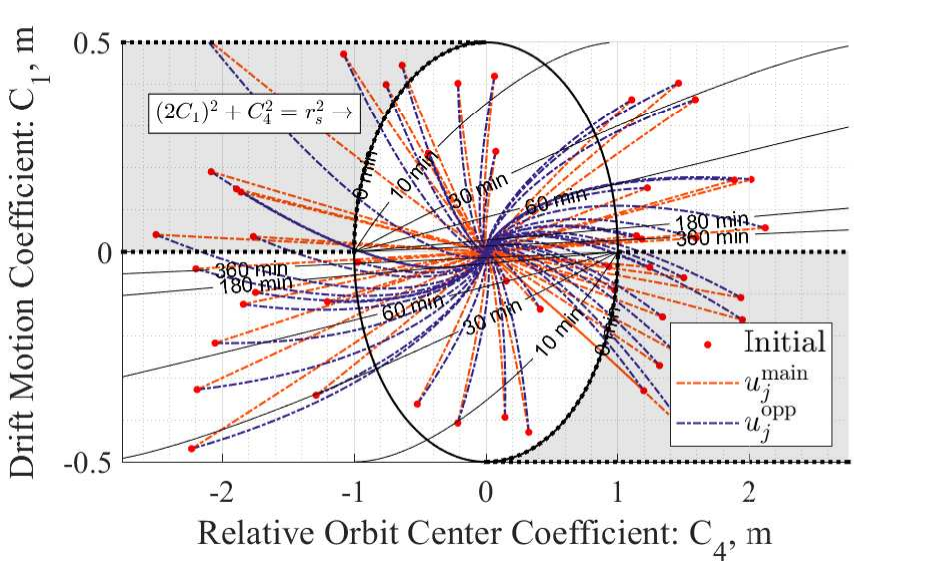}
    \subcaption{Large $C_1$ conditions: Connectable period transition.}
  \end{minipage}\\
  \begin{minipage}[b]{1\singlecolumnwidth}
    \centering
\includegraphics[width=1.075\columnwidth]{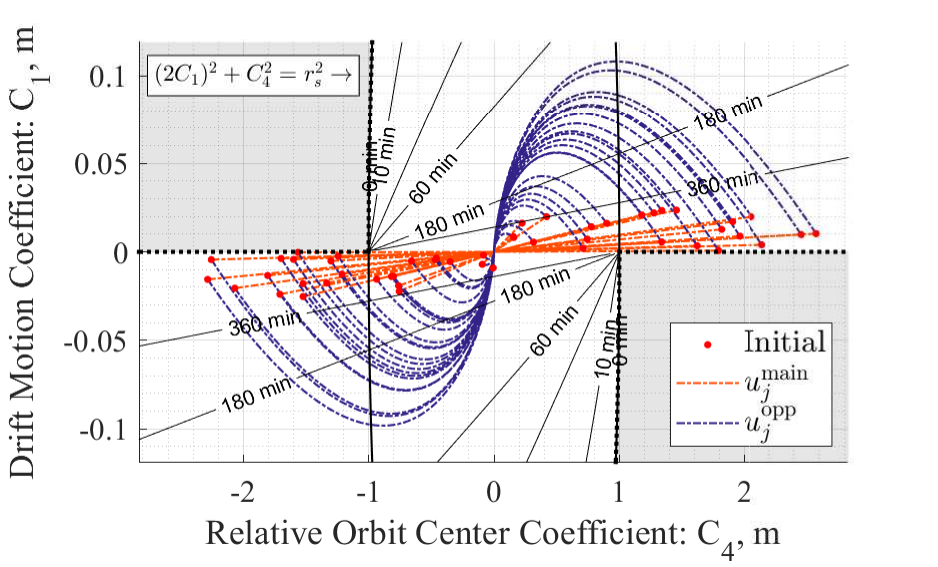}
    \subcaption{Small $C_1$ conditions: Connectable period transition.}
  \end{minipage}\\
  \begin{minipage}[b]{1\singlecolumnwidth}
    \centering
      \hspace*{-0.25cm}%
\includegraphics[angle=90, width=1.1\columnwidth]{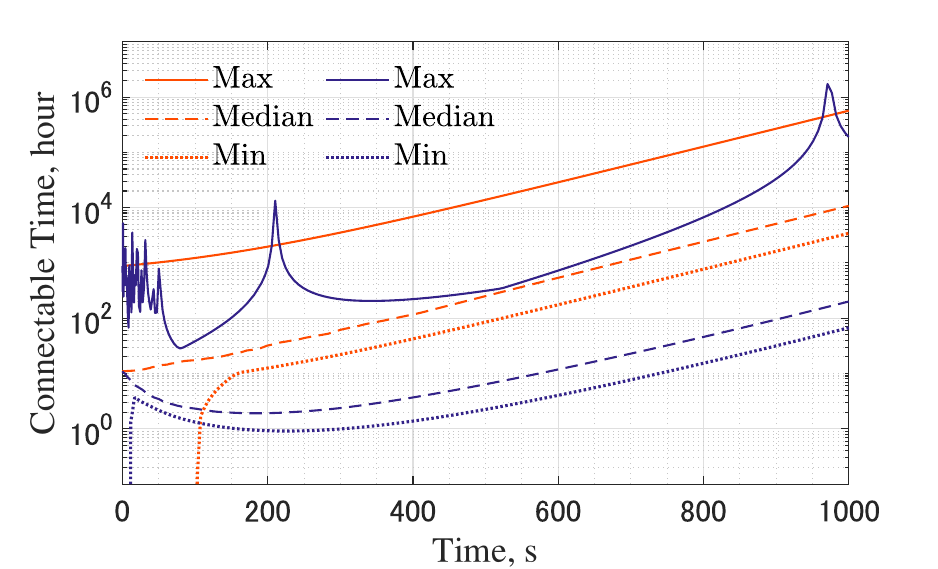}
    \subcaption{Small $C_1$ conditions: Connectable period history.}
\label{Small_C_1_conditions_Connectable_period_history}
  \end{minipage}
  \caption{Relaxed connectable period transition 
  in Eq.~(\ref{T_conn})
  for 50 satellites in subsection~\ref{Connectable_Period_Maximization}. 
Figure~\ref{Small_C_1_conditions_Connectable_period_history} shows their maximum, median, and minimum values.
  }  \label{Relaxed_connectable_period_for_two_cases}
\end{figure}
\subsection{Two Formation Control Results into Different Planes}
\label{Networked_Control_Results_with_Constant_Gains}
We finally validated the distance-based formation control results of a 100-satellite swarm in two types of orbital planes. The proposed controller ${u}^{\mathrm{main}}_j$ in Eq.~(\ref{networked_input_i}) is applied to $J_2$ perturbed orbital dynamics in Eq.~(\ref{J2 orbital dynamics}) to avoid linearization error. To initialize $r_{\mathrm{o}(0)}$ within the connectable region of $S/C_{\mathrm{ref}}$, the initial orbital parameters from $S/C_{\mathrm{ref}}$ are determined as relatively dense initial conditions as follows:
$$
|C_1|\leq r_s/2,\quad |C_4|\leq r_s.
$$  
Almost all $r_{\mathrm{o}}(0)$ are randomly initialized within the connectable region of $S/C_{\mathrm{ref}}$, i.e., $r_{\mathrm{o}}^2=(2C_1)^2+C_4^2\leq r_s^2$. Note that this condition allows some satellites to exist outside this region, or the communication period is limited to relatively short periods. The remaining orbital plane information of swarm plane $\Phi(\Theta_{P},\Theta_{z-xy})$ and following $r_{xyd}$ are as follows:
$$
\begin{aligned}
\angle\Theta_{1}\mathrm{: }&&\Theta_P = 30^\circ,\ \Theta_{Z-XY} = 00^\circ\ &\Rightarrow r_{xyd}\approx \frac{r_{\mathrm{avg}}}{2+4.00e^{-5}}\\
\angle\Theta_{2}\mathrm{: }&&\Theta_P = 40^\circ,\ \Theta_{Z-XY} = 50^\circ\ &\Rightarrow  r_{xyd}\approx \frac{r_{\mathrm{avg}}}{2+1.12e^{-1}}
\end{aligned}
$$ 
First, $\Theta_{P,Z-XY}$ realizes almost constant relative distances of the stable formation as shown in subsection \ref{J2_Relative_Orbital_Parameters}.

The control results by ${u}^{\mathrm{main}}_j$ in Eq.~(\ref{networked_input_i}) are summarized in Figs.~\ref{Networked:Formation_shape} and \ref{Networked:Formation_Orbital_Parameters}. Fig.~\ref{Networked:Formation_shape} shows the initial or final formation shape for $\angle\Theta_{1}$ case in (a,b) and $\angle\Theta_{2}$ case in (c,d). The red line indicates each satellite's position, which represents its attitude, and the black dotted line indicates mutual interactions. The five dotted circles indicate a user-defined plane $\Phi(\Theta_{P},\Theta_{z-xy})$. These confirm convergence into a user-defined plane before $t=120h$ under time-varying $J_2$ geopotential gravity effects. 

Fig.~\ref{Networked:Formation_Orbital_Parameters} shows the averaged $J_2$ orbital parameter history of 99 satellites with respect to the S/C$_{\mathrm{ref}}$ for $\angle\Theta_{1}$ case. The results shown in Fig.~\ref{Networked:Formation_Orbital_Parameters} a) are related to the first target type and drift elimination derived in subsection~\ref{swarm_plane}, and they summarize the connectable period transition. These findings show that all $C_{1,4}$ converge to the origin, which means that concentric orbits are formed. The index $\hat{T}_{\mathrm{conn}}$ provides the satellite deployment behavior for limited $\overline{f}$. All $C_1$ tends to decay exponentially, although $C_4$ with $C_1C_4\leq 0$ tends to diverge in the initial period. This supports their hierarchical relationship in Eq.~(\ref{dot_C14}), and the exponential stability of $C_4$ is guaranteed only for sufficiently small $C_1$ under limited input. Lower gain $k_A$ also causes non-convergence due to the initial $C_4$ divergence. The results of Fig.~\ref {Networked:Formation_Orbital_Parameters} b) are related to the second and third target types, namely, orbit plane and distance control derived in subsection~\ref{swarm_plane}. The orbital distance $r_{xy,z}$ converges to a constant around the desired orbital distance that satisfies $W_sE^\top[C_2]=0$, as proven in theorem~\ref{theorem_networked_controller}. Moreover, orbital angles $\Theta_{P,Z-XY}$ converge into the desired angles based on the hierarchical target strategy in subsection~\ref{swarm_plane}.
\begin{figure}[tb!]
\centering
  \begin{minipage}[b]{0.48\singlecolumnwidth}
    \centering
\includegraphics[width=1.05\columnwidth]{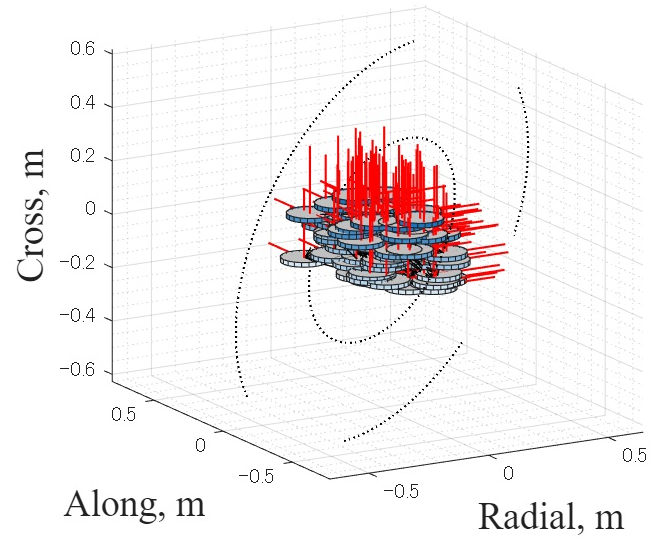}
    \subcaption{Initial states: $\angle\Theta_{1}$.}
  \end{minipage}
  \begin{minipage}[b]{0.48\singlecolumnwidth}
    \centering
    \includegraphics[width=1.05\columnwidth]{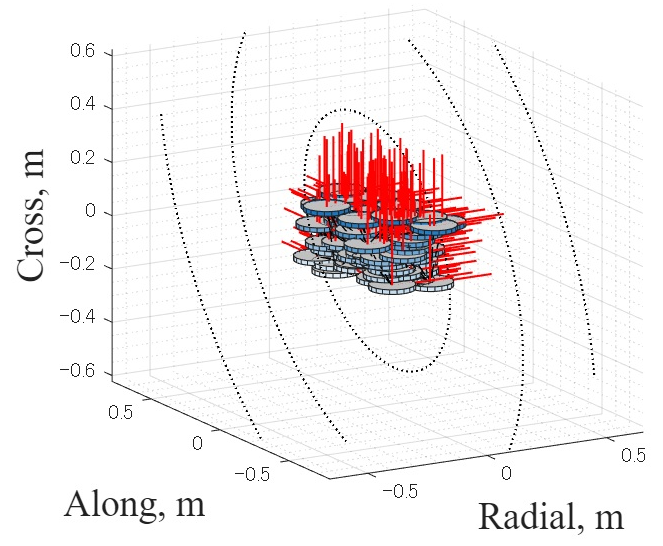}
    \subcaption{Initial states: $\angle\Theta_{2}$.}
  \end{minipage}\\
  \begin{minipage}[b]{0.5025\singlecolumnwidth}
    \centering
\includegraphics[width=1.05\columnwidth]{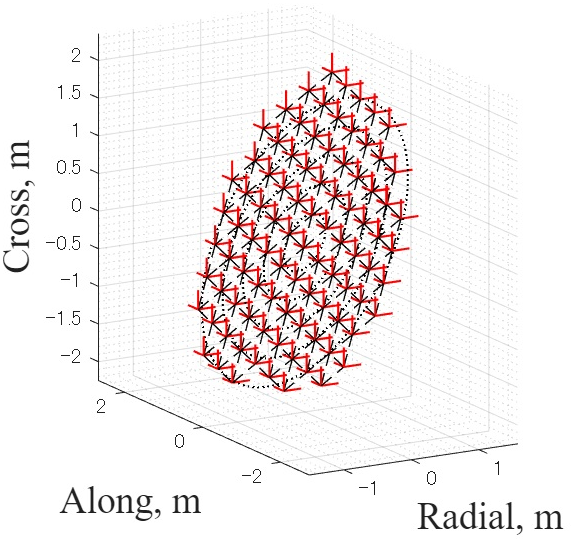}
    \subcaption{LVLH states (120h): $\angle\Theta_{1}$.}
  \end{minipage}
  \begin{minipage}[b]{0.4675\singlecolumnwidth}
    \centering
\includegraphics[width=1.05\columnwidth]{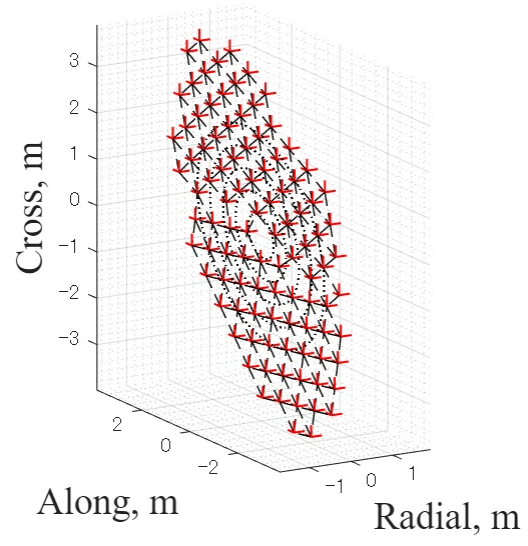}
    \subcaption{LVLH states (120h): $\angle\Theta_{2}$.}
  \end{minipage}\\
  \begin{minipage}[b]{0.485\singlecolumnwidth}
    \centering
\includegraphics[width=1.025\columnwidth]{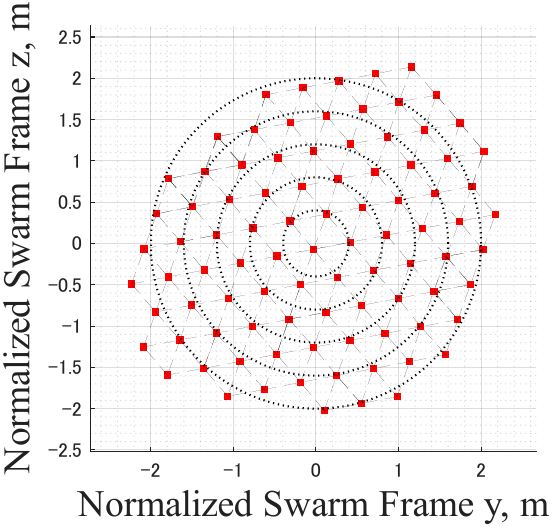}
    \subcaption{$\{\hat{S}\}$ states (120h): $\angle\Theta_{1}$.}
  \end{minipage}
  \begin{minipage}[b]{0.485\singlecolumnwidth}
    \centering
    \includegraphics[width=1.025\columnwidth]{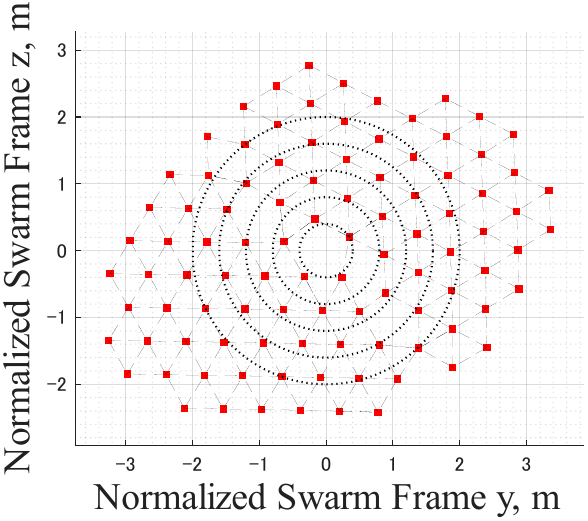}
    \subcaption{$\{\hat{S}\}$ states (120h): $\angle\Theta_{2}$.}
  \end{minipage}
  \caption{One-hundred satellite formation control results at orbital frames $\{\mathcal{O}\}$ and $\{\hat{S}\}$ in subsection~ \ref{Networked_Control_Results_with_Constant_Gains}.}
  \label{Networked:Formation_shape}
\end{figure}
\begin{figure}[tb!]
\centering
  \begin{minipage}{\singlecolumnwidth}
    \centering
    \hspace*{-0.02\columnwidth} 
\includegraphics[width=1.085\columnwidth]{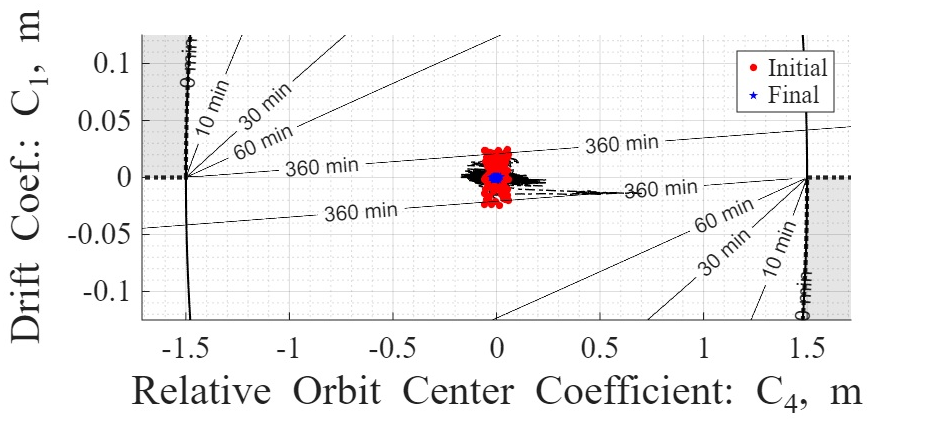}
    \subcaption{Connectable period transition of 99 satellites.}
  \end{minipage}\\
   \begin{minipage}{\singlecolumnwidth}
    \centering
    \hspace*{-0.05\columnwidth} 
\includegraphics[width=1.125\columnwidth]{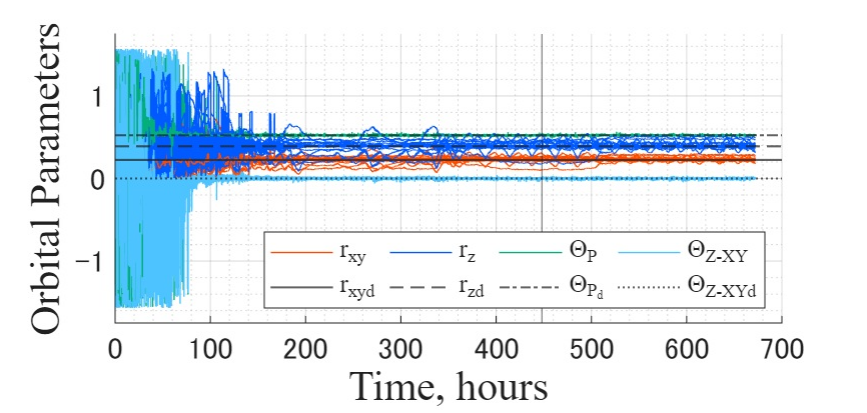}
    \subcaption{Averaged $J_2$ orbital parameters for the reference satellite}
  \end{minipage}
  \caption{Averaged $J_2$ orbital parameters in Eq.~(\ref{definition_C}) from the reference for 100 satellites under time-varying $J_2$ geopotential gravity effects in subsection~\ref{Networked_Control_Results_with_Constant_Gains}: Formation shape of the orbital frame $\{\mathcal{O}\}$.}
  \label{Networked:Formation_Orbital_Parameters}
\end{figure}
\section{Conclusion}
\label{Conclusion}
This study investigated a formation control strategy for achieving large-scale and scalable space structures using a satellite swarm. We primarily considered fuel-free actuation, such as magnetic field interaction and differential aerodynamic drag, and derived averaged $J_2$ orbital parameters to attenuate the drift motion and control the periodic motion. This provided a decentralized deployment controller in a user-defined orbital plane at a specific distance. Since fuel-free actuation is typically satellite distance-dependent, i.e., the control performance degrades for large separations, our controller minimizes drift during unexpected communication loss under unstable orbital dynamics. These results lower the technical barriers of large-scale space systems, thereby improving system performance in terms of resolution, communication speed, and effective isotropic radiated power.
\section*{Appendices}
\label{Appendices}
\subsection{Escape, Connectable, and Relaxed Connectable Time} 
\label{Escape_Connectable_Time_Relaxation}
We derived the representative times for a deputy satellite to exit from a controllable region for controller design. Based on the analytical solutions in subsection \ref{Averaged_J_2_Relative_Orbital_Parameters}, we provide the escape time $T_{\mathrm{out}}\in\mathbb{R}$, connectable time $T_{\mathrm{conn}}$, and “relaxed” connectable time $\hat{T}_{\mathrm{conn}}(C_{1,4})$. We first denote the strict definition of "escape time" and "connectable time" under the distance-based constraint.
\begin{definition}
For parameters $C_{1-6}(t_0)$ of a given relative orbit at initial time $t_0$, time $t_s=t_0+{T}_{\mathrm{out}}(C_{1-6})$ corresponds to the time when the distance between two satellites coincides with $r_s$, where
$$
{T}_{\mathrm{out}}(C_{1-6})=\left\{t_s \geq t_0\ |\ \|r(t_s)\|^2 = r_s^2 \right\}-t_0
$$
and $r(t)=[x(t),y(t),z(t)]^\top$ of the closed-form solutions in Eq.~(\ref{CWsol}). Here, the initial escape time $\underline{T}_{\mathrm{out}}(C_{1-6})=\min({T}_{\mathrm{out}})$ and final escape time $\overline{T}_{\mathrm{out}}(C_{1-6})=\max({T}_{\mathrm{out}})$ are estimated at initial time $t_0$ through the closed-form solutions as follows:
$$
\begin{aligned}
\underline{T}_{\mathrm{out}}(C_{1-6})&=\mathrm{min}\left\{t \geq t_0\ |\ \|r(t)\|^2 = r_s^2 \right\}-t_0\\
\overline{T}_{\mathrm{out}}(C_{1-6})&=\mathrm{max}\left\{t \geq t_0\ |\ \|r(t)\|^2 = r_s^2 \right\}-t_0
\end{aligned}
$$
\end{definition}
\begin{definition}
Connectable time $T_{\mathrm{conn}}(C_{1-6})$ is determined as follows:
$$
T_{\mathrm{conn}}(C_{1-6})=\sum\left\{t_0\leq t \leq \overline{T}_{\mathrm{out}}(C_{1-6})\ |\ \|r(t)\|^2 \leq r_s^2 \right\}.\\
$$
which satisfies $\underline{T}_{\mathrm{out}}\leq T_{\mathrm{conn}}\leq \overline{T}_{\mathrm{out}}$ since $\overline{T}_{\mathrm{out}}$ includes the duration for the deputy satellite to escape from the controllable region.
\end{definition} 

One drawback of these strict estimations is the computational cost associated with some numerical techniques, which can be prohibitive for online execution by small satellites. Therefore, we derive the lower bound $\hat{\underline{T}}_{\mathrm{out}}(C_{1,4},r_{xy,z})$, such that $\hat{\underline{T}}_{\mathrm{out}}(C_{1,4},r_{xy,z})\leq {\underline{T}}_{\mathrm{out}}(C_{1-6})$. We consider the following ellipse as a virtual relative trajectory $r_v(C_{1-6},t_k,\theta)$:
$$
r_v(C_{1-6},t_k,\theta)=r_{\mathrm{o}}(t_k)+
\begin{bmatrix}
   r_{xy}\sin{(\theta + \theta_{xy})}/c_+\\
    2r_{xy}\cos{(\theta + \theta_{xy})}/c_-\\
    r_z \sin{(\theta+\theta_z)}
\end{bmatrix}
$$
where the fixed time is $t_k$, arbitrary phase is $\theta \in[0,2\pi)$, and origin of the ellipse is $r_{\mathrm{o}}(t)$ in Eq.~(\ref{CWsol}). We define the coordinate transformation $C^{E/O}$ such that the ellipse lies on the $X-Y$ plane, with its minor axis aligned along the $X$-axis and the origin of the ellipse $[X_o;Y_o;Z_o]=C^{E/O}r_{\mathrm{o}}(t_k)$ as introduced in \cite{chang2011novel}. Here, to derive $\hat{\underline{T}}_{\mathrm{out}}(C_{1,4},r_{xy,z})$, we consider that the inter-satellite distance takes the maximum value. This occurs when $[1,{dY}/{dX}]^\top$ and $[1,{dY_o}/{dX_o}]^\top$ intersect orthogonally, where ${dY}/{dX}$ is the gradient of the tangent at $[X;Y]$ on the ellipse and ${dY_o}/{dX_o}$ is the vectors that joins the origin and the arbitrary point on the ellipse:
$$
\frac{dY}{dX}=-\frac{X-X_o(t_k)}{Y-Y_o(t_k)} \frac{l_\mathrm{max}^2}{l_\mathrm{min}^2},\quad\frac{dY_o}{dX_o}=\frac{Y}{X}.
$$
Therefore, we can derive the constraint of $\hat{\underline{T}}_{\mathrm{out}}$:
\begin{equation*}
\label{T_out}
\hat{\underline{T}}_{\mathrm{out}}(C_{1,4},r_{xy,z})=\mathrm{min}\left\{t \geq t_0\ \left|\
\begin{aligned}
&X^2+Y^2 = r_s^2\\
&1+\frac{dY}{dX}\frac{dY_o}{dX_o}=0
\end{aligned}
\right\}\right.
\end{equation*}
This can be simplified to a cubic equation whose solutions are up to four actual points and could be provided analytically. Since this $\hat{\underline{T}}_{\mathrm{out}}(C_{1,4},r_{xy,z})$ still presents high computational costs, we define a relaxed time function in Eq.~(\ref{T_conn}).
\begin{lemma}
For a given $C_{1,4}$, consider an arbitrary $r_{s0}\in\mathbb{R}_+$ such that $r_{s0}^2\geq (2C_1)^2$ and an associated $\hat{T}_{0}(r_{s0})\geq0$. Then, the following is satisfied for $\delta C_{4}=(C_4-3\omega C_1\hat{T}_{0})$:
\begin{equation*}
\label{bound_Tconn2}
\begin{aligned}
\hat{T}_{\mathrm{conn}}(C_{1,4})
&=\hat{T}_0(r_{s0})+\frac{-|\delta C_4|+\sqrt{r_s^2-(2C_1)^2}}{\epsilon_2|C_1|}\\
\end{aligned}
\end{equation*}
\end{lemma}
\begin{proof}
Substituting Eq.~(\ref{T_conn}) into $\mathrm{sgn}(C_1\delta C_4(\hat{T}_0- \hat{T}_{\mathrm{conn}}))$ yields $-\mathrm{sgn}( \hat{T}_{0}- \hat{T}_{\mathrm{conn}})$; thus, $\mathrm{sgn}(C_1\delta C_4)\leq 0$. Applying this and $\delta C_4$ for Eq.~(\ref{T_conn}) results in Eq.~(\ref{bound_Tconn2}). 
\end{proof}
\subsection{Centralized Grouping on a Normalized Swarm Frame}
\label{two_networking_alg}
We present the criteria for the satellite grouping method based on the Normalized Swarm Frame and derive a multi-leader network $\mathcal{D}^{\mathfrak{ml}}\left(\mathcal{V}, \mathcal{E}^{\mathfrak{ml}}\right)$. To reduce computational costs, we assume that the neighboring satellites $\mathcal{N}_j$ are updated mainly by the closeness of the relative distance and remove overlapping interactions. Then, we limit $\mathcal{N}_j$ to inclusion in the Delaunay triangulation (DT) graph \cite{lee1980two,schwab2021distributed,guinaldo2024distributed}, where the edges form triangles, and the circumcircle contains no points in $\mathcal{V}$ in its interior. Here, updating DT based on the Euclid distance leads to the non-constant $\mathcal{N}_j$ since even the stable state $p_d(t)$ is periodic, as shown in Eq.~(\ref{desired_position}). Then, DT is derived by $p_d(t)$ in the $\{\hat{\mathcal{S}}\}$ frame to obtain the time-invariant $\mathcal{N}$, which leads to the convergence of the distance-based formation. In this study, we calculate DT in a centralized manner. For comparison, please refer to the decentralized DT estimation \cite{schwab2021distributed} and 3D surfaces \cite{guinaldo2024distributed} based on a given initial DT. Since the number of neighbor candidates can be larger than $\Delta$, we sort them by $\hat{T}_{\mathrm{conn}}$ and pair them with candidates presenting a small $\hat{T}_{\mathrm{conn}}$. These candidates are considered unsafe to ensure a certain amount of communication time.   

The specific centralized groping algorithm \ref{alg:Centralized_Multi_Leader_Grouping} is used to derive a multi-leader network $\mathcal{D}^{\mathfrak{ml}}\left(\mathcal{V}, \mathcal{E}^{\mathfrak{ml}}\right)$, which includes leader agent set $\mathcal{V}_l$ and edges for $j$th group $\mathcal{E}_j^{\mathfrak{ml}}$. Again, this will be replaced by an arbitrary method. We can achieve multi-leader selection through the optimization problem, such as \cite{clark2012leader}. Certain crucial points of the grouping algorithm are noted, which can ensure the stability of the distance-based control. First, multiple groups do not share a common edge for stable control allocation, i.e., the control of each edge should be realized within a single group. Second, $\overline{n}_f$ should be greater than three because the edge number of a DT is $[3(|\mathcal{V}|-1)-e(t)]$ \cite{lee1980two}, where $e(t)$ is the number of vertices on the boundary $e(t)\in[3, N]$. Third, the outer edges of the distance between neighboring satellites become non-uniform. Then, removing these connections at the outer edges of the DT under certain conditions may result in a uniform overall formation. Fourth, the edges of $\mathcal{N}$ should not cross each other to guarantee graph rigidity for convergence stability.

To ensure that the maximum value of the neighbor number remains less than or equal to $\Delta$, we define multi-leader digraph $\mathcal{D}^{\mathfrak{ml}}(\mathcal{V}, \mathcal{E}^{\mathfrak{ml}})$ to derive $\mathfrak{g}_{l}$ and edge $(j, k)\in\mathcal{E}^{\mathfrak{ml}}$ describes the hierarchical relationships of $k\in\mathcal{V}_l$ and $j\in\mathcal{V}/k$, i.e., the $j$-th satellite follows $k$-th satellite.
\begin{definition}[Multi-Leader Digraph]
\label{multi_leader_graph}
For given $\mathfrak{g}_{l}$, we define the multi-leader digraph $\mathcal{D}^{\mathfrak{ml}}(\mathcal{V}, \mathcal{E}^{\mathfrak{ml}})$ and edge $(j, k)\in\mathcal{E}^{\mathfrak{ml}}$ describes the hierarchical relationships of $k\in\mathcal{V}_l$ and $j\in\mathcal{V}/k$, i.e., the $j$-th satellite follows $k$-th one.
$$
\mathfrak{g}_{l}=\{l, f\in \mathcal{V}\  |\ (f, l)\in\mathcal{E}_l^{\mathfrak{ml}} \},\quad\ \mathcal{E}^{\mathfrak{ml}}=\mathcal{E}_1^{\mathfrak{ml}} \cap
\ldots \cap \mathcal{E}_{N}^{\mathfrak{ml}}
$$
where $\mathfrak{g}_{l}=l$ and $\mathcal{E}_{l}^{\mathfrak{ml}}=\{\}$ means that the $l$th leader does not have followers. Its non-weighted, non-symmetric adjacency matrix $A(\mathcal{D}^{\mathfrak{ml}})$ is defined as follows:
$$
[A(\mathcal{D}^{\mathfrak{ml}})]_{j k}=\begin{cases}1
& \text { if }(k, j) \in \mathcal{E}^{\mathfrak{ml}} \\ 0 & \text { otherwise }\end{cases}
$$
Here, the sum of each $j$th row $\sum_{k=1}^{n} [A(\mathcal{D}^{\mathfrak{ml}})]_{jk}$ and $k$th column $\sum_{j=1}^{n} [A(\mathcal{D}^{\mathfrak{ml}})]_{jk}$ indicate the number of $j$th followers and $k$-th following, respectively:
$$
\sum_{k=1}^{n} [A(\mathcal{D}^{\mathfrak{ml}})]_{jk}\leq \overline{n}_{f\leftarrow l},\quad \sum_{j=1}^{n} [A(\mathcal{D}^{\mathfrak{ml}})]_{jk}\leq  \overline{n}_{l\leftarrow f}
$$
These bounds include $\Delta$ as a function of $\overline{n}_{l\leftarrow f},\overline{n}_{f\leftarrow l}$ 
e.g. $\Delta(\overline{n}_{l\leftarrow f},\overline{n}_{f\leftarrow l})=(\overline{n}_{l\leftarrow f}+1)\overline{n}_{f\leftarrow l}\ $ if $\ \overline{n}_{l\leftarrow f}=\overline{n}_{f\leftarrow l}=2$.
\end{definition}
\begin{algorithm}[tb!]
\caption{Centralized Multi-Leader Grouping for Connected Graphs or Each Connected Subgraph}
\label{alg:Centralized_Multi_Leader_Grouping}
\begin{algorithmic}[1] 
\State \textbf{Inputs: }1) $\mathcal{V}=\{1,\ldots ,N\}$, 2) ${r}_j^{\hat{s}}\in\mathbb{R}^3$ by Eq.~(\ref{conversion_for_grouping}), 3) $r_s\in\mathbb{R}$, 4) $\overline{n}_f\in\mathbb{R}$, 5) $\mathcal{V}_l=\{\}$, $\mathcal{E}_j^{\mathfrak{ml}}=\{\}$
\State \textbf{Outputs: }$\mathcal{D}^{\mathfrak{ml}}\left(\mathcal{V}, \mathcal{E}^{\mathfrak{ml}}\right)$, $\mathcal{V}_l$, $\mathcal{E}_l^{\mathfrak{ml}}$
\State Extract $\Phi(\Theta_{P},\Theta_{z-xy})$ component, ${r}_{j\mathrm{yz}}^{\hat{s}}={r}_j^{\hat{s}}(2:3)$
\State Use ${r}_{j\mathrm{yz}}^{\hat{s}}$ to make adjacency matrix $A_{\mathcal{DT}}$ and its degree matrix $D_{\mathcal{DT}}$ from Delaunay triangulation 
\State Derive index $k_{\mathrm{deg}}$ in order of degree of $D_{\mathcal{DT}}$ 
\For{$k_l=k_{\mathrm{deg}}(m_1\in[1\dots N])$}
\State Set $\mathfrak{g}_{k_l}=k_l$, add $k_l$ into $\mathcal{V}_l$, and derive index $k_{\hat{T}_{\mathrm{conn}}}$ in order of $\hat{T}_{\mathrm{conn}}^{-1}$ for $k_l$th agent
\For{$k_f=k_{\hat{T}_{\mathrm{conn}}}(m_2\in[1\dots N])$}
\If{$A_{\mathcal{DT}}(k_l,k_f)$ $\&\&$ $\|\bm{r}_{k_l}-\bm{r}_{k_f}\|\leq r_s$}
\State Add $k_f$ into $\mathrm{g}_{k_l}$
\If{size of $\mathfrak{g}_{k_l}$ is $\overline{n}_f\in\mathbb{R}$} break
\EndIf
\EndIf
\EndFor
\For{$(k_1\in g_{k_l},k_2\in g_{k_l})$}
\If{$A_{\mathcal{DT}}(k_1,k_2)$}
\State Add $(k_{1,2},k_{2,1})$ into $\mathcal{E}_{k_l}^{\mathfrak{ml}}$
\EndIf
\EndFor
\EndFor
\end{algorithmic}
\end{algorithm}
\subsection{Proof of Theorem~\ref{theorem_networked_controller}: Baseline Distance-based Orbital Stabilizer}
\begin{proof}
\label{proof_theorem_networked_controller}
We compute a smooth collective potential function $\mathcal{V}_{\mathrm{all}}\in\mathbb{R}$ as a candidate of the Lyapunov function:
\begin{equation}
\label{potential_function}
\begin{aligned}
&V_{\mathrm{all}}(t,\varrho)=
\frac{k_A}{2}
\left\|
\begin{bmatrix}
\hat{E}^\top[2C_1]\\
\sqrt{\gamma_A} E^\top [C_4-\psi C_1]
\end{bmatrix}
\right\|^2+V_{\delta}\\
&V_{\delta}=\sum_{i}\sum_{j\in\mathcal{N}_i}{k_{Bij}}\int_{\|r_{xyd}\|_\sigma}^{\|\overline{r}_{xyij}\|_\sigma}s(\mathsf{r}) \sqrt{1+\epsilon\|\mathsf{r}\|^2}\mathrm{d}\mathsf{r}
\end{aligned}
\end{equation}
where the $\sigma$-norm $\|z\|_\sigma=(\sqrt{1+\epsilon\|z\|^2}-1)/\epsilon$ with a fixed parameter $\epsilon \in(0,1)$. 
\begin{equation}
\label{drdt_2}
\begin{aligned}
\frac{\mathrm{d}\overline{r}_{xyjk}^2}{\mathrm{d}t}
=&2k_0
\begin{bmatrix}
\tau(-\varsigma +\frac{c_-}{2})C_{42}\\
(\frac{\varrho}{2}-c_+ )C_{13}+\frac{-\varsigma \psi\tau}{2}C_{42}
\end{bmatrix}
\begin{bmatrix}
U_x\\
U_y
\end{bmatrix}\\
&-2\omega_{xy}\tau\left(\frac{\epsilon_2 \varsigma}{\omega_{xy}}-\varrho\right)C_{42}^\top C_1\\
&-2\omega_{xy}C_{13}\left\{\varsigma (C_4-\psi C_1)+(\tau-1) C_{42}\right\}\\
\end{aligned}
\end{equation}
where $C_{42}=(\varsigma (C_4-\psi C_1)+C_2)$, $C_{13}=(\varrho C_1+C_3)$. Let $L_s(t)\in\mathbb{R}^{N\times N}$ be defined as a weighted Laplacian matrix: 
$$
L_{s(ij)}=\left\{
\begin{aligned}
&-k_{Bij}s(\overline{r}_{xyij}),\ &&\mathrm{if}\ j \in \mathcal{N}_i \\
&\sum_{j \in N_i}k_{Bij}s(\overline{r}_{xyij}),\ &&\mathrm{if}\ i=j \\
&0,\ &&\mathrm{else}
\end{aligned}.
\right.
$$
Then, the time derivative of $V_{\mathrm{all}}$ is
$$
\begin{aligned}
\dot{V}_{\mathrm{all}}
=&2
\begin{bmatrix}
\mathfrak{e}_4+k_1\tau \mathfrak{e}_2\\
\frac{k_A}{2}\hat{E}\hat{E}^\top[2C_1]+\frac{\psi}{2} E \mathfrak{e}_4\\
{(\frac{\varrho}{2}-c_+ )}\mathfrak{e}_3+\frac{k_1(-\varsigma \psi\tau)/{2}}{(-\varsigma +\frac{c_-}{2})}\mathfrak{e}_2
\end{bmatrix}^\top
\begin{bmatrix}
k_0 E^\top U_x\\
k_0 U_y\\
k_0 E^\top U_y
\end{bmatrix}\\
&+\frac{2\epsilon_2}{k_A}
\begin{bmatrix}
\mathfrak{e}_4\\
\mathfrak{e}_2
\end{bmatrix}^\top
\begin{bmatrix}
    1&\frac{2\omega_{xy}\varsigma}{\gamma_A \epsilon_2}\\
    \frac{c_-}{2c_+}\frac{(\frac{\varrho}{2}-\frac{\epsilon_2 \varsigma}{2\omega_{xy}})}{(-\varsigma +\frac{c_-}{2})}&O\\
\end{bmatrix}
\begin{bmatrix}
\mathfrak{e}_1\\
\mathfrak{e}_3
\end{bmatrix}\\
&+\frac{k_A}{2}
[2C_1] \left\{\frac{\mathrm{d}}{\mathrm{d}t}\left(\hat{E}\hat{E}^\top\right)\right\}[2C_1]\\
\end{aligned}
$$
where we use $\frac{\mathrm{d}\|r_{xyij}\|_\sigma}{\mathrm{d}t}=r_{xyij}\dot{r}_{xyij}/\sqrt{1+\epsilon\|\mathsf{r}_{xyij}\|^2}$, $k_1=(c_-\epsilon_2)/(4c_+\omega_{xy})$, and $\mathfrak{e}_{1,2,3,4}$ are defined as the following edge vectors:
$$
\begin{bmatrix}
    \mathfrak{e}_1\\\mathfrak{e}_3
\end{bmatrix}
=
\begin{bmatrix}
\frac{k_A}{2}E^\top [2C_1]\\    
W_sE^\top C_{13}
\end{bmatrix}
,\ 
\begin{bmatrix}
\mathfrak{e}_4\\\mathfrak{e}_2
\end{bmatrix}
=
\begin{bmatrix}   
-\frac{k_A\gamma_A }{2} E^\top[C_4-\psi C_1]\\
\frac{(-\varsigma +\frac{c_-}{2})}{k_1}W_s E^\top  C_{42}
\end{bmatrix}.
$$

We first consider our main controller $u_j$ in~(\ref{networked_input_i}) for $k_B\neq0$ and $k_B=0$. The matrix formulation of $u_j$ in~(\ref{networked_input_i}) is as follows:
$$
\begin{bmatrix}
    U_x\\
    U_y
\end{bmatrix}
=\frac{1}{k_0}(I_2\otimes E)
\begin{bmatrix}
    -\gamma_B^2 \mathfrak{e}_4+f_{0}\mathfrak{e}_1 -\frac{\gamma_B^2}{k_1}\mathfrak{e}_2\\
    -\mathfrak{e}_1+g_{0}\mathfrak{e}_2\\
\end{bmatrix}
$$
where $\psi,k_1,f_{0},g_{0}$ are presented in Eq.~(\ref{coefficient_condition}). We note that ${\overline{r}}_{xy(0,0,1,0)}^2=r_{\mathrm{xy}}^2$, $\overline{r}_{xy(2c_+,c_-/2,1,0)}^2=\overline{x}^2+(\overline{y}/2)^2$, and ${\overline{r}}_{xy({c_-\epsilon_2}/{(2\omega_{xy})},0,1,0)}^2=({\dot{\overline{x}}}/{\omega_{xy}})^2+({\dot{\overline{y}}}/{(2\omega_{xy})})^2$. It is worth noting that the above satisfies $U_j =E[\cdot]$, which indicates that $[C_i] = [C_i] -\frac{\mathbf{1}^\top C_i}{n}\mathbf{1}$. By selecting $(\varrho,\varsigma,\tau,\psi)=(2c_+,0,1,0)$ and $\hat{E}=E$ and applying this input, the following time derivative of $V_{\mathrm{all}}$ is derived:
$$
\begin{aligned}
\dot{V}_{\mathrm{all}}
&=2
\begin{bmatrix}
\mathfrak{e}_4+k_1\mathfrak{e}_2\\
\mathfrak{e}_1+\frac{\psi}{2}\mathfrak{e}_4\\
\left(\mathfrak{e}_4 +\mathfrak{e}_2\right)
\end{bmatrix}^\top
\begin{bmatrix}
k_0E^\top U_x\\
k_0E^\top U_y\\
\frac{\epsilon_2}{k_A}\mathfrak{e}_1
\end{bmatrix}
=-2
x^\top
\begin{bmatrix}
\gamma_B^2 L_e &
\gamma_BK_0\\
\gamma_BK_0& \gamma_B^2 L_e 
\end{bmatrix}
x
\end{aligned}
$$
where $K_0=(\frac{\lambda_0}{2} L_e-\frac{\epsilon_2}{2k_A}I)$, $x=[\mathfrak{e}_1/\gamma_B;\mathfrak{e}_4+\mathfrak{e}_2]$. Moreover, the definition of $\psi$ in Eq.~(\ref{coefficient_condition}) is the solution of 
$$
\frac{\psi}{2}\left((1-k_1)f_{0}-\frac{\psi}{2}\right)=\gamma_B \left(k_1+\frac{1}{k_1}-2\right)=\gamma_B \frac{(k_1-1)^2}{k_1}
$$ where $c\triangleq(k_1+1/k_1-2)=\frac{(k_1-1)^2}{k_1}$. The following relationship for arbitrary vectors $\mathsf{a}\in\mathbb{R}^{|\mathcal{E}|}$ and $\mathsf{b}\in\mathbb{R}^{|\mathcal{N}|}$ is then derived:  
$$
\begin{aligned}
\mathsf{a}^\top   E^\top \mathsf{b}&
= \mathsf{a}^\top    V_+(V_+^\top V_+)\Sigma_+^\top U_+^\top \mathsf{b}
= \mathsf{a}^\top V_+V_+^\top E^\top \mathsf{b}\\
\end{aligned}
$$
where $V_+\in\mathbb{R}^{\smash{|\mathcal{E}|\times |\mathcal{V}|-1}}$ in subsection~\ref{Algebraic_Graph_Theory}. By applying $\mathsf{a}=\mathfrak{e}_4 +\mathfrak{e}_2$ and $E^\top \mathsf{b}=\mathfrak{e}_1$ for the above relationship into $\dot{V}_{\mathrm{all}}$, we obtained the following:
$$
\begin{aligned}
\dot{V}_{\mathrm{all}}
&=-2\gamma_B
\hat{x}^\top
\begin{bmatrix}
\gamma_B D_+ &
\frac{\lambda_0}{2} D_+ -\frac{\epsilon_2}{2k_A}I\\
*& \gamma_B D_+ 
\end{bmatrix}\hat{x}\\
&=-\gamma_B\|V_+^\top[\mathfrak{e}_1/\gamma_B\pm (\mathfrak{e}_4+\mathfrak{e}_2)]\|^2_{P_{\pm}}
\end{aligned}
$$
where $\hat{x}=(I_2\otimes V_+^\top)x$ and $P_{\pm}\triangleq \gamma_B D_+\pm(\frac{\lambda_0}{2} D_+-\frac{\epsilon_2}{2k_A}I)$. Since the definition in Eqs.~(\ref{coefficient_condition}) and (\ref{definition_d}) guarantees $P_{\pm}\succeq 0$, we obtain $\|V_+^\top[\mathfrak{e}_1/\gamma_B\pm (\mathfrak{e}_4+\mathfrak{e}_2)]\|\rightarrow 0$ as $t\rightarrow 0$; thus, $\|V_+^\top[\mathfrak{e}_1/\gamma_B]\|\rightarrow 0$ and $\|V_+^\top[ (\mathfrak{e}_4+\mathfrak{e}_2)]\|\rightarrow 0$. 
Then, $\|V_+^\top \mathfrak{e}_1\|\rightarrow 0$ shows $ that 
\|U_+^\top [C_1]\|\rightarrow 0$, which indicates that $[C_1]$ asymptotically converges to the eigenspace associated with the zero eigenvalue, i.e., $[C_1] \rightarrow \frac{\mathbf{1}^\top C_1}{|\mathcal{V}|}\mathbf{1}$ or $E^\top[C_1] \rightarrow 0$ by the definition in subsection~\ref{Algebraic_Graph_Theory}. The equation of $E^\top\dot{C}_{1}$ with $  
\|U_+^\top [C_1]\|\rightarrow 0$ implies that $E^\top U_y\rightarrow 0$ and $\mathfrak{e}_2,W_s E^\top  C_{2}\rightarrow 0$. Finally, $\|V_+^\top[ (\mathfrak{e}_4+\mathfrak{e}_2)]\|\rightarrow 0$ shows that
$$
U_+^\top [C_4]\rightarrow \frac{c_-}{k_1 k_A\gamma_A}V_+^\top W_s E^\top  C_{2}\rightarrow 0
$$
where $U_+^\top [C_4]$ is the eigenspace associated with the zero eigenvalue and $E^\top [C_4]\rightarrow 0$ and $E^\top U_x\rightarrow 0$. For the case of $k_B=0$, we obtain globally asymptotically stable results for $E^\top[C_{1,4}]$ for arbitrary gains. Integrating $[\dot{C}_1]=-\frac{k_A}{2}L[C_1]$ yields an exponentially stable result: 
$\|\mathfrak{e}_1\|\leq \|\mathfrak{e}_{1(0)}\|e^{-(k_A\lambda_2/2)t}$. We define $\hat{E}\hat{E}^\top = \alpha(t)L$, where $\alpha(t)>0$ is as follows:
$$
\alpha(t)\triangleq\alpha_0e^{\frac{k_A\lambda_2 t}{2}}>0\quad\text{s.t.}\quad\alpha_0\geq \frac{2|-\lambda_0|\|L\|_2 +{2\epsilon_2}/{k_A}}{\lambda_2\|\mathfrak{e}_{1(0)}\|/\sup_t \|\mathfrak{e}_{4(t)}\|}
$$
This provides $\dot{V}_{\mathrm{all}}\leq -2\gamma_B^2\mathfrak{e}_4^\top L \mathfrak{e}_4+\dot{V}_{14}$ where
$$
\begin{aligned}
&\dot{V}_{14}=\left(\left(2(-\lambda_0) L+\frac{2\epsilon_2}{k_A}I
\right)\mathfrak{e}_4-\left(2\alpha L- \frac{2\dot{\alpha}}{k_A}I\right)\mathfrak{e}_1\right)^\top\mathfrak{e}_1\\
&\leq\left(\left(2|-\lambda_0|\|L\|_2 +\frac{2\epsilon_2}{k_A}\right)
\sup_t \|\mathfrak{e}_4(t)\|-\lambda_2\|\mathfrak{e}_{1(0)}\|
\alpha_0\right)\|\mathfrak{e}_1\|.
\end{aligned}
$$
Based on the inequality of $\alpha_0$, we obtain $\dot{V}_{14}\leq 0$ and $\dot{V}_{\mathrm{all}}\leq -2\gamma_B^2\mathfrak{e}_4^\top L \mathfrak{e}_4\leq 0$. This guarantees that $\mathfrak{e}_4\rightarrow 0$ as $t\rightarrow\infty$, and thus, $E^\top C_4\rightarrow 0$ along with the exponentially stable results of $E^\top C_1$.

We next consider the subcontroller $\hat{u}_j$ in Eq.~(\ref{sub_networked_input_i}), and its matrix formulation is given as follows:
\begin{equation*}
\begin{bmatrix}
{U}_x(\hat{u}_j)\\
{U}_y(\hat{u}_j)
\end{bmatrix}
=
\begin{bmatrix}
\frac{k_A\gamma_A}{2k_0}L[C_4] -\frac{c_-}{2k_0} L_s [C_2]\\
-\frac{k_A}{2k_0}L[2C_1]
+\frac{\gamma_A \epsilon_2}{2k_0}L[C_4]
-\frac{\epsilon_2}{k_A k_0}\frac{c_-}{2k_1}L_sC_{2}
\end{bmatrix}.
\end{equation*}
Then, by applying $(\varrho,\varsigma,\tau,\psi)=(2c_+,0,1,0)$, $\hat{E}=I$, and ${U}_{x,y}(\hat{u}_j)$, we obtained the following:
$$
\begin{aligned}
\dot{V}_{\mathrm{all}}
&=2k_0^2
\begin{bmatrix}
-\frac{k_A\gamma_A }{2k_0} L[C_4]+\frac{c_-}{2k_0}L_s  C_{2}\\
U_y-\frac{\gamma_A\epsilon_2}{2k_0} L[C_4]+\frac{\epsilon_2}{k_Ak_0}\frac{c_-}{2k_1}L_s  C_{2}
\end{bmatrix}^\top
\begin{bmatrix}
 U_x\\
\frac{k_A}{2k_0}[2C_1] 
\end{bmatrix}\\
&=-\frac{k_A^2}{2}\|E^\top [2C_1]\|^2-2k_0^2 \|{U}_x\|^2\\
\end{aligned}
$$
which indicates that both $\|E^\top [C_1]\|$ and ${U}_x(\hat{u}_j)$ converge to 0 as $t \rightarrow \infty$. Consequently, the equation for $E^\top\dot{C}_{1,4}$ implies that $E^\top U_y\rightarrow 0$ and $E^\top [C_4]\rightarrow\mathrm{const}$. The definitions of $U_x$ and $U_y$ reveal that $E^\top [C_4]\rightarrow 0$ and $W_s E^\top C_2\rightarrow 0$. For the case of $k_B=0$, it is evident that the $E^\top [C_4]\rightarrow 0$ and $E^\top [C_4]\rightarrow 0$ as $t\rightarrow\infty$.

We finally mention the convergence results of $\overline{r}_{xyjk}$ for arbitrary $jk$-th satellites. When we achieve $U_{xy}\rightarrow 0$ and $E^\top [C_{1,4}]\rightarrow 0$, we obtained $\dot{\overline{r}}_{xyjk}=0$ in Eq.~(\ref{drdt_2}). This indicates that $\overline{r}_{xyjk}$ takes a constant value and $[C_{2,3}]$ takes a periodic value. These satisfy $W_sE^\top[C_2]=0$ for $k_B\neq 0$. 

Finally, we show the stability proof for the $z$-motion for two satellites, which can be extended to the $N$ satellite case. For the $z$-axis controller, consider the Lyapunov candidate function $V_{56}=(C_5-C_{5r})^2/2$, which satisfies $\dot{C}_{5r}=-\omega_{xy}C_{6d}+\cos\theta_0\hat{d}_y$, where $\hat{d}_y=-k_1\frac{c_-}{2}\frac{\epsilon_2}{k_A}\frac{k_0}{2}d_b$, the constant gains $k_z\in\mathbb{R}_+$, $C_{5d}$ are defined in Eq.~(\ref{r_xy_and_r_z}), and $C_{5r}={C}_{5d}+{c_-k_1\cos\theta_0}\delta{C}_4/2-2c_+k_1\sin\theta_0{C}_1$. Substituting $u_{\mathrm{safe}}$ into this time derivative yields $\dot{V}_{56}=-2k_{z} \delta C_5^2\leq 0$ for $\hat{d}_y=d_z=0$. Then, $(C_5,C_6)$ asymptotically converges to $(C_{5d},C_{6d})$, driving the satellites to the desired orbital plane as $C_{1,4}\rightarrow 0$.
\end{proof}
\section*{Funding Sources}
\noindent
This research was supported by JST's Next-Generation Challenging Research Program JPMJSP2106 and the JAXA, Space Strategy Fund GRANT Number (JPJXSSF24MS09003), Japan. 
\section*{Acknowledgments}
\noindent 
The authors would like to thank Hiraku Sakamoto, Seang Shim, and Hayate Tajima for their technical discussion.
\bibliographystyle{new-aiaa}
\bibliography{referencesdownload}
\end{document}